\documentclass{lmcs} 
\pdfoutput=1

\usepackage{lastpage}
\lmcsdoi{21}{1}{23}
\lmcsheading{}{\pageref{LastPage}}{}{}%
{Jul.~04,~2024}{Mar.~11,~2025}{}

\usepackage[utf8]{inputenc}

\usepackage{amsmath}
\usepackage{amssymb}
\usepackage{amsthm}
\usepackage{mathabx}
\usepackage{tikz}
\usepackage{tikz-qtree}
\usetikzlibrary{shapes.geometric, arrows}
\usetikzlibrary{decorations.markings}
\tikzset{negated/.style={
		decoration={markings,
			mark= at position 0.5 with {
				\node[transform shape] (tempnode) {$\backslash$};
			}
		},
		postaction={decorate}
	}
}
\usetikzlibrary{hobby} 
\usepackage{virginialake}
\usepackage{bussproofs}

\usepackage{hyperref}
\usepackage[capitalise]{cleveref}
\crefformat{enumi}{(#2#1#3)}
\crefname{propC}{Proposition}{Propositions}

\renewcommand{\epsilon}{\varepsilon}
\renewcommand{\phi}{\varphi}

\newcommand{\emptylist}{\epsilon}

\renewcommand{\vec}[1]{\mathbf{#1}}

\newcommand{\IH}{\mathit{IH}}

\newcommand{\storageone}{}
\newcommand{\storagetwo}{}
\newcommand{\storagethree}{}
\newcommand{\storagefour}{}

\newtheorem{theorem}[thm]{Theorem}
\newtheorem{proposition}[theorem]{Proposition}
\newtheorem{lemma}[theorem]{Lemma}
\newtheorem{corollary}[theorem]{Corollary}
\newtheorem{observation}[theorem]{Observation}

\theoremstyle{definition}
\newtheorem{definition}[theorem]{Definition}
\newtheorem{remark}[theorem]{Remark}
\newtheorem{example}[theorem]{Example}

\newtheorem{notation}[theorem]{Notation}

\newcommand{\mon}{\mathbf{m}}
\newcommand{\co}{\mathit{co}}
\newcommand{\N}{\mathbf{N}}
\newcommand{\Logspace}{\mathbf{L}}
\newcommand{\NL}{\N\Logspace}

\newcommand{\NC}[1]{\mathbf{NC}^{#1}}

\newcommand{\Ptime}{\mathbf{P}}
\newcommand{\NP}{\mathbf{NP}}
\newcommand{\coNP}{\co\NP}

\newcommand{\AC}[1]{\mathbf{AC}^{#1}}

\newcommand{\bool}{\{0,1\}}
\newcommand{\Nat}{\mathbb{N}}

\newcommand{\monclo}[1]{m(#1)}
\newcommand{\posclo}[1]{#1^+}

\newcommand{\exact}[2]{\mathrm{Ex}^{#1}_{#2}}
\newcommand{\thresh}[2]{\mathrm{Th}^{#1}_{#2}}

\newcommand{\sat}[1]{\vDash_{#1}}
\newcommand{\notsat}[1]{\nvDash_{#1}}

\newcommand{\ee}[2]{e_{#1#2}}

\newcommand{\cnot}{\neg}
\newcommand{\dual}[1]{\overline{#1}}
\newcommand{\cand}{\wedge}
\newcommand{\cimp}{\supset}

\newcommand{\posterm}[1]{\mathrm{Conj}(#1)}

\newcommand{\dec}[3]{#1#2#3}
\newcommand{\posdec}[3]{#1#2(#1\lor #3)}
\newcommand{\posdecsmall}[3]{{\scriptstyle #1} { #2} {\scriptstyle (\hspace{-.05em} #1 \hspace{-0.13em} \lor \hspace{-0.13em} #3 \hspace{-.05em} )}}

\newcommand{\extiff}{\mbox{\Large $\, \leftrightarrow \, $}}

\newcommand{\ex}[2]{e^{#1}_{#2}}
\newcommand{\thr}[2]{t^{#1}_{#2}}

\newcommand{\exextaxs}[2]{\mathcal{E}^{#1}_{#2}}
\newcommand{\thrextaxs}[2]{\mathcal{T}^{#1}_{#2}}

\newcommand{\refthr}[4]{[\hspace{-.05em} \posdecsmall{#3}{\thr{#1}{#2}}{#4} \hspace{-.05em} ]}

\newcommand{\seqar}{\mbox{\Large $\, \rightarrow \, $}}
\newcommand{\dseqar}{\extiff}
\newcommand{\lk}{\mathsf{LK}}
\newcommand{\mlk}{\mathsf{MLK}}
\newcommand{\ext}{\mathsf{e}}

\newcommand{\lndt}{\mathsf{LNDT}}

\newcommand{\elndt}{\ext\lndt}

\newcommand{\pos}[1]{#1^{+}}
\newcommand{\poselndt}{\pos\elndt}
\newcommand{\negposelndt}{\pos\elndt_{-}}
\newcommand{\tposelndt}[1]{\pos\elndt_{#1}(P)}
\newcommand{\lefrul}[1]{#1\text{-}l}
\newcommand{\rigrul}[1]{#1\text{-}r}

\newcommand{\id}{\mathsf{id}}
\newcommand{\cntr}{\mathsf{c}}
\newcommand{\wk}{\mathsf{w}}
\newcommand{\cut}{\mathsf{cut}}

\newcommand{\pij}[2]{p_{#1 #2 }}
\newcommand{\PPi}[1]{\mathbf{p}_{#1}}
\newcommand{\PPj}[1]{\mathbf{p}^\intercal_{#1}}

\newcommand{\php}[1]{\mathsf{PHP}_{#1}}
\newcommand{\lphp}[1]{\mathsf L \php{#1}}
\newcommand{\rphp}[1]{\mathsf R \php{#1}}

\newcommand{\ppdel}[1]{\vec p_{#1}}

\newcommand{\negtrans}[1]{#1^{-}}
\newcommand{\ttrans}[2]{#2^{#1}}

\theoremstyle{plain}

\begin{document}

\title[Proof complexity of positive branching programs]{Proof complexity of positive branching programs}
\titlecomment{{\lsuper*} This is an extended version of a conference paper published at \emph{Computability in Europe '22}.}
\thanks{The alphabetically first author has been supported by a UKRI Future Leaders Fellowship, \emph{Structure vs Invariants in Proofs}, project number
MR/S035540/1.}

\author[A.~Das]{Anupam Das\lmcsorcid{0000-0002-0142-3676}}
\author[A.~Delkos]{Avgerinos Delkos}

\address{University of Birmingham, UK}	

\email{a.das@bham.ac.uk, a.delkos@bham.ac.uk}

\begin{abstract}
  \noindent 
  We investigate the proof complexity of systems based on positive branching programs, 
i.e.\ non-deterministic branching programs (NBPs) where, for any 0-transition between two nodes, there is also a  1-transition. 
Positive NBPs compute monotone Boolean functions, just like negation-free circuits or formulas, but constitute a positive version of (non-uniform) $\NL$, rather than $\Ptime$ or $\NC 1$, respectively.

The proof complexity of NBPs was investigated in previous work by Buss, Das and Knop, using extension variables to represent the dag-structure, over a language of (non-deterministic) decision trees, yielding the system $\elndt$.
Our system $\poselndt$ is obtained by restricting their systems to a positive syntax, similarly to how the `monotone sequent calculus' $\mlk$ is obtained from the usual sequent calculus $\lk$ by restricting to negation-free formulas.

Our main result is that $\poselndt$ polynomially simulates $\elndt$ over positive sequents. Our proof method is inspired by a similar result for $\mlk$ by Atserias, Galesi and Pudl\'ak, that was recently improved to a bona fide polynomial simulation via works of Je\v r\'abek and Buss, Kabanets, Kolokolova and Kouck\'y.
Along the way we formalise several properties of counting functions within $\poselndt$ by polynomial-size proofs and, as a case study, give explicit polynomial-size proofs of the propositional pigeonhole principle. 
\end{abstract}

\maketitle

\section{Introduction}
\emph{Proof complexity} is the study of the complexity of formal proofs.
This pursuit is fundamentally tied to open problems in computational complexity, in particular due to the Cook-Rechow theorem \cite{DBLP:journals/jsyml/CookR79}: $\coNP = \NP$ if and only if there is a propositional proof system (suitably defined) that has polynomial-size proofs of each propositional tautology.
This has led to what is known as `Cook's program' for separating $\Ptime$ and $\NP$: find superpolynomial lower bounds for stronger and stronger systems until a general method is found (see, e.g., \cite{Bus12:tow-np-p-via-prf-comp-search,Kra19:Cook-Reckhow-definition}).

Systems of interest in proof complexity are often motivated by analogous considerations from circuit complexity. 
For instance {bounded depth} systems restrict proofs to formulas with a limit on the number of alternations between $\lor$ and $\cand $ in its formula tree, i.e.\ $\AC 0 $ formulas.
Indeed, H\r{a}stad's famous lower bound techniques for $\AC 0$ \cite{DBLP:journals/iandc/HaastadWWY94} have been lifted to the setting of proof complexity, yielding lower bounds for a propositional formulation of the {pigeonhole principle} \cite{DBLP:conf/stoc/BeameIKPPW92} in bounded depth systems via a refined version of the {switching lemma}. 

\emph{Monotone} proof complexity is motivated by another famous lower bound result, namely Razborov's lower bounds on the size of $\neg$-free circuits \cite{razborov1985,Raz85:lower-bds-mon-comp-log-permanent} (and similar ones for formulas \cite{KarWig88:st-conn-super-log-depth,RazWig90:mNC1-vs-NC1}).
 In this regard, there has been much investigation into the negation-free fragment of Gentzen's sequent calculus, called $\mlk$ \cite{DBLP:journals/mlq/AtseriasGG01,DBLP:journals/jcss/AtseriasGP02,DBLP:journals/apal/Jerabek11a,BKKK17}. 
 Atserias, Galesi and Pudl\'ak showed in
 \cite{DBLP:journals/jcss/AtseriasGP02} a quasipolynomial simulation of $\lk$ by $\mlk$ on $\neg$-free sequents by formalising an elegant counting argument using quasipolynomial-size negation-free counting formulae.
 This has recently been improved to a polynomial simulation by an intricate series of results \cite{DBLP:journals/jcss/AtseriasGP02,DBLP:journals/apal/Jerabek11a,BKKK17}, solving a question first posed in \cite{PudBuss}.

 However, 
 note the contrast with bounded depth systems: restricting negation has different effects on computational complexity and on proof complexity.
 
 In this work we address a similar question for the setting of \emph{non-deterministic branching programs}.
These are believed to be more expressive than Boolean formulas, in that they are the non-uniform counterpart of non-deterministic log-space ($\NL$), as opposed to $\NC 1$.
They have recently been given a proof theoretic treatment in \cite{DBLP:conf/csl/BussDasKnop20}, in particular addressing proof complexity.
We work within that framework, only restricting ourselves to formulas representing \emph{positive} branching programs.

Positive (or `monotone') branching programs have been considered several times in the literature, e.g.\ \cite{GrigniSnipser,KarWig93:on-span-programs}. 
They are identical to Markov's `relay-diode bipoles' from \cite{Markov62:minimal-relay-diode-bipoles}.
\cite{GrigniSnipser,grigni1991structure} give a general way of making a non-deterministic model of computation `positive'; in particular, a non-deterministic branching program is positive if, whenever there is a $0$-transition from a node $u$ to a node $v$, there is also a $1$-transition from $u$ to $v$. 
As in the earlier work \cite{DBLP:conf/csl/BussDasKnop20} we implement such a criterion by using disjunctions to model nondeterminism.
As far as we are aware, there is no other work investigating the proof complexity of systems based on positive/monotone branching programs.

\subsection{Contribution}
We present a formal proof calculus $\poselndt$, reasoning with formula-based representations of positive branching programs, by restricting the calculus $\elndt$ from \cite{DBLP:conf/csl/BussDasKnop20} appropriately.
We consider the `positive closures' of well-known polynomial-size `ordered' BPs (OBDDs) for counting functions, and show that their characteristic properties admit polynomial-size proofs in $\poselndt$.

As a case study, we show that these properties can be used to obtain polynomial-size proofs of the propositional pigeonhole principle, by adapting an approach of \cite{DBLP:journals/mlq/AtseriasGG01} for $\mlk$.

Our main result is that $\poselndt$ in fact polynomially simulates $\elndt$ over positive sequents.
For this we again use representations of positive NBPs for counting and small proofs of their characteristic properties.
At a high level we adapt the approach of \cite{DBLP:journals/jcss/AtseriasGP02}, 
but there are several additional technicalities specific to our setting.
In particular, we require bespoke treatments of negative literals in $\elndt$ and (iterated) substitutions of representations of positive NBPs into other positive NBPs. 

\subsection{Preliminary version}
This is an extended version of the conference paper \cite{DasDelkos22}. The current version expands that work by full proofs of all results, as well as extra narrative and examples throughout.

\subsection{A terminological convention}
Throughout this work, we shall reserve the words `monotone', `monotonicity' etc.\ for \emph{semantic} notions, i.e.\ as properties of Boolean functions.
For (non-uniform) models of computation such as formulas, branching programs, circuits etc., we shall say `positive' for the associated \emph{syntactic} constraints, e.g.\ negation-freeness for the case of formulas or circuits.
While many works simply say `monotone' always, in particular \cite{grigni1991structure,GrigniSnipser}, let us note that the distinction we make is employed by several other authors too, e.g.\ \cite{AjtGur87:mon-vs-pos,LauSchSte96:on-pos-p,LauSchSte98:pos-vers-polytime,DasOit18:posfp}.

\section{Preliminaries on proof complexity and branching programs}

In this section we will recall some preliminary content from \cite{DBLP:conf/csl/BussDasKnop20}. The reader familiar with that work can safely omit this section, though they should take note of our conventions in the definition of the system $\elndt$, cf.~\cref{differences-from-original-elndt}.

Throughout this work we will use a countable set of \emph{propositional variables}, written $p,q$ etc., and \emph{Boolean constants} $0$ and $1$.

An \emph{assignment} is just a map $\alpha$ from propositional variables to $\bool$. For all intents and purposes we may assume that they have finite support, e.g.\ nonzero only on variables occurring in a formula or proof.
We extend an assignment $\alpha$ to the constants in the natural way, setting $\alpha(0) = 0 $ and $\alpha(1) =1$.

A \emph{Boolean function} is just a map from (finitely supported) assignments to $\bool$.

\subsection{Proof complexity}
In proof complexity, a \emph{propositional proof system} is just a poly\-nomial-time function $P$ from $\Sigma^*$ to the set of propositional tautologies, where $\Sigma$ is some finite alphabet.
Intuitively, the elements $\sigma\in \Sigma^*$ code proofs in the system, while $P$ itself is an efficient `proof-checking' algorithm that verifies that $\sigma$ is indeed a correctly written proof; if so $P$ returns its conclusion, i.e.\ the theorem it proves. 
If not, it just returns $1$, by convention.

The significance of this definition is due to the following result from \cite{DBLP:journals/jsyml/CookR79}:

\begin{theorem}
[Cook-Reckhow]
There is a propositional proof system with polynomial-size proofs of each tautology if and only if $\coNP = \NP$.
\end{theorem}

\noindent
In practice, this `Cook-Reckhow definition' of a propositional proof system covers all well-studied proof systems for propositional logic, under suitable codings. 
We shall refrain from giving any of these codings explicitly in this work, as is standard for proof complexity.
However, let us point out that the systems we consider routinely admit polynomial-time proof checking in the way described above, and so indeed constitute formal propositional proof systems.

See \cite{Kra95:ba-pl-ct,CooNgu10:log-found-prf-comp,Kra19:prf-comp} for more comprehensive introductions to proof complexity.

\subsection{Non-deterministic branching programs}
A (non-deterministic) \emph{branching program} (NBP) is a (rooted) directed acyclic graph $G$ with two distinguished \emph{sink} nodes, $0$ and $1$, such that:
\begin{itemize}
\item $G$ has a unique root node, i.e.\ a unique node with in-degree $0$.
    \item Each non-sink node $v$ of $G$ is labelled by a propositional variable.
    \item Each edge $e$ of $G$ is labelled by a constant $0$ or $1$.
\end{itemize}

A \emph{run} of a NBP $G$ on an assignment $\alpha$ is a maximal path beginning at the root of $G$ consistent with $\alpha$. 
I.e.\ at a node labelled by a propositional variable $p$ the run must follow an edge labelled by $\alpha(p) \in \bool$.

We say that $G$ \emph{accepts} $\alpha$ if there is a run on $\alpha$ reaching the sink $1$.
We may extend $\alpha$ to a map from all NBPs to $\bool$ by setting $\alpha(G) = 1 $ if $G$ accepts $\alpha$, and $\alpha(G) = 0$ otherwise.
In this way,
each NBP \emph{computes} a unique Boolean function $\alpha \mapsto \alpha(G)$.

A comprehensive introduction to (variants of) branching programs and their underlying theory can be found in, e.g., \cite{Weg00:bps-and-bdds}.

\begin{example}
[2-out-of-4 Exact]
\label{example-2-4-exact}
The 2-out-of-4 Exact function, which returns 1 when precisely two of its four arguments are 1, is computed by the branching program in \cref{fig:exact-obdd-4-2}. 
$0$-edges are indicated dotted and $1$-edges are indicated solid, a convention that we adopt throughout this work. Formally, each $0$-leaf corresponds to the same sink.

Note that this program is deterministic: there is exactly one 0-edge and one 1-edge outgoing from each non-sink node. 
It is also \emph{ordered}: all the variables appear in the same order in each path.
Its semantics may be verified by checking that every path leading to the 1-sink has exactly two 1-edges and vice versa.
\end{example}

   \begin{figure}[t]
    
\begin{tikzpicture}[scale=1.08, auto,swap]
\foreach \pos/\name/\disp in {
  {(0,4)/1/$p_1$}, 
  {(-1,3)/2/$p_2$},
  {(1,3)/3/$p_2$}, 
  {(-2,2)/4/$p_3$}, 
  {(0,2)/5/$p_3$},
  {(2,2)/6/$p_3$}, 
  {(-3,1)/8/$p_4$},
  {(-1,1)/10/$p_4$},
  {(1,1)/11/$p_4$},
  {(3,1)/12/$p_4$},
  {(-4,-0)/15/$0$},
  {(-2,-0)/16/$0$},
  {(0,-0)/17/$1$},
  {(+2,-0)/18/$0$},
  {(+4,-0.008)/19/$0$}}
\node[minimum size=20pt,inner sep=0pt] (\name) at \pos {\disp};

    \draw [->][thick,dotted](1) to (2);
    \draw [->][thin](1) to (3);
    
    \draw [->][thick,dotted](2) to (4);
    \draw [->][thin](2) to (5);
     
    \draw [->][thick,dotted](3) to (5);
    \draw [->][thin](3) to (6);

     \draw [->][thick,dotted](4) to (8);
     \draw [->][thin](4) to (10);
      \draw [->][thick,dotted](5) to (10);
    \draw [->][thin](5) to (11);
     \draw [->][thick,dotted](6) to (11);
     \draw [->][thin](6) to (12);
     \draw [->][thin](8) to (16);
     \draw [->][thick,dotted](8) to (15);
      \draw [->][thin](10) to (17);
     \draw [->][thick,dotted](10) to (16); \draw [->][thin](11) to (18);
     \draw [->][thick,dotted](11) to (17); \draw [->][thin](12) to (19);
     \draw [->][thick,dotted](12) to (18); 
\end{tikzpicture} 
\caption{A branching program computing the 2-out-of-4 Exact function. $0$-edges are indicated dotted, and $1$ edges are indicated solid.}
\label{fig:exact-obdd-4-2}\end{figure}

\subsection{Representation of NBPs by extended formulas}

Since we will be working in formal proof systems, we shall use a natural representation of NBPs by `formulas with extension', just like in \cite{DBLP:conf/csl/BussDasKnop20}.
For this,
we shall make use of \emph{extension variables} $e_0, e_1, e_2, \dots$ in our language, disjoint from the set of propositional variables.

An \emph{extended non-deterministic decision tree} formula, or \emph{eNDT} formula, written $A,B$ etc., is generated as follows:
\[
A,B, \dots \quad ::= \quad 0 \quad \mid \quad 1 \quad \mid \quad p  \quad \mid \quad \dec A p B \quad \mid \quad A \lor B \quad \mid \quad e_i
\]
Formulas of form $\dec A p B$ are called \emph{decisions}, and intuitively express ``if $p$ then $B$ else $A$''.

As usual, we often omit external brackets of formulas and write long disjunctions without internal brackets, under associativity.
The \emph{size} of a formula $A$, written $|A|$, is the number of symbols occurring in $A$.

\begin{remark}
[Distinguishing extension variables]
Note that we formally distinguish extension variables from propositional variables.

This is for the same technical reasons as in \cite{DBLP:conf/csl/BussDasKnop20} we must not allow extension variables to be decision variables, i.e.\ we forbid formulas of the form $\dec A {e_i} B$. 
If we did allow this then we would be able to express all Boolean circuits succinctly, whereas the current convention ensures that we only express NBPs. 
\end{remark}

The semantics of (non-extended) NDT formulas under an assignment is standard.

With extension variables, however, the interpretation is parametrised by a set of \emph{extension axioms}, allowing extension variables to `abbreviate' more complex formulas.

\begin{definition}
[Extension axioms]
\label{extension-axioms-definition}
A \emph{set of extension axioms} $\mathcal A$ is a set of the form $ \{ e_i \extiff A_i \}_{i<n} $, where each $A_i$ may only contain extension variables among $e_0, \dots, e_{i-1}$.  
\end{definition}

\begin{definition}
[Semantics of eNDT formulas]
\label{dfn:sat-endt-wrt-extax}
\emph{Satisfaction} with respect to a set of extension axioms $\mathcal A = \{ e_i \extiff \mathcal A_i \}_{i<n}$, written $\sat{\mathcal A}$, is a (infix) binary relation between assignments and formulas over $e_0, \dots, e_{n-1}$ defined as follows:
\begin{itemize}
\item $\alpha \notsat {\mathcal A} 0$ and $\alpha \sat {\mathcal A} 1$.
    \item $\alpha \sat {\mathcal A} p$ if $\alpha(p) = 1$.
    \item $\alpha \sat {\mathcal A} A \lor B$ if $\alpha \sat {\mathcal A} A$ or $\alpha \sat {\mathcal A} B$.
    \item $\alpha \sat {\mathcal A} \dec A p B$ if either $\alpha(p)=0$ and $\alpha \sat {\mathcal A} A$, or $\alpha(p)=1$ and $\alpha \sat {\mathcal A} B$. 
    \item $\alpha \sat {\mathcal A} e_i$ if $\alpha \sat {\mathcal A} A_i$.
\end{itemize}
\end{definition}

\begin{example}
[2-out-of-4 Exact, revisited]
\label{ex:exact24-ext-axioms}
Recall \cref{example-2-4-exact} and the branching program from \cref{fig:exact-obdd-4-2}. Under the semantics above, we may represent this branching program by the formula $\ee 1 1 $ under the following extension axioms:
\begin{equation}
\label{eq:exact24-ext-axioms}
    \begin{array}{r@{\ \extiff \ }l}
    \ee 1 1 & \dec {\ee 2 1 } {p_1}{\ee 2 2} \\
    \ee 2 1 & \dec {\ee 3 1 } {p_2}{\ee 3 2} \\
    \ee 2 2  & \dec {\ee 3 2 } {p_2}{\ee 3 3}
\end{array}
\qquad
\begin{array}{r@{\ \extiff \ }l}
    \ee 3 1  & \dec {\ee 4 1 } {p_3}{\ee 4 2} \\
    \ee 3 2  & \dec {\ee 4 2 } {p_3}{\ee 4 3} \\
    \ee 3 3  & \dec {\ee 4 3 } {p_3}{\ee 4 4}
\end{array}
\qquad
\begin{array}{r@{\ \extiff \ }l}
    \ee 4 1  & \dec 0 {p_4} 0 \\
    \ee 4 2  & \dec 0 {p_4} 1 \\
    \ee 4 3  & \dec 1 {p_4} 0 \\
    \ee 4 4  & \dec 0 {p_4} 0
\end{array}
\end{equation}
Each $\ee i j$ represents the $j$\textsuperscript{th} node (left to right) on the $i$\textsuperscript{th} row (top to bottom). 
Note that, in order to strictly comply with the subscripting condition on extension axioms, we may identify $\ee i j $ with $e_{4(4-i)+j}$.
\end{example}

Note that the notion $\sat{\mathcal A}$ is indeed well-defined, thanks to the subscripting conditions on sets of extension axioms: intuitively, each $e_i$ abbreviates a formula containing only extension variables among $e_0, \dots, e_{i-1}$, and so on.
Indeed this well-foundedness gives rise to an induction scheme:
\begin{remark}
[$\mathcal A$-induction]
\label{A-induction}
Given a set of extension axioms $\mathcal A = \{ e_i \extiff A_i \}_{i<n} $ we may define a strict partial order $<_\mathcal A$ on formulas over $e_0,\dots, e_{n-1}$ by:
\begin{itemize}
    \item $p <_\mathcal A \dec A p B$ and $A <_\mathcal A \dec A p B$ and $B <_\mathcal A \dec A p B$.
    \item $A <_\mathcal A A \lor B$ and $B <_\mathcal A A \lor B$.
    \item $A_i <_\mathcal A e_i$, for each $i<n$.
\end{itemize}
Notice that $<_\mathcal A$ is indeed well-founded by the condition that each $A_i$ must contain only extension variables among $e_0, \dots, e_{i-1}$.
Thus we may carry out arguments and make definitions by induction on $<_\mathcal A$, which we shall simply refer to as `$\mathcal A$-induction'. 
\end{remark}

We can now see Definition~\ref{dfn:sat-endt-wrt-extax} above of $\sat {\mathcal A}$ as just a definition by $\mathcal A$-induction.
In this way, fixing some set of extension axioms $\mathcal A = \{e_i \extiff A_i\}_{i<n} $, each eNDT formula $A$ over $e_0, \dots, e_{n-1}$ computes a unique Boolean function $f: \alpha \mapsto 1 $ if $\alpha \sat {\mathcal A} A$.
In this case, we may say that $A$ \emph{computes} $f$ \emph{with respect to} $\mathcal A$.

Since many of our arguments will be based on $\mathcal A$-induction, let us make the following observation for complexity matters:
\begin{observation}
[Complexity of $\mathcal A$-induction]
\label{complexity-of-A-induction}
Let $\mathcal A = \{ e_i \extiff A_i\}_{i<n} $ be a set of extension axioms and $A$ contain only extension variables among $e_0, \dots, e_{n-1}$.
Then $| \{ B<_\mathcal A A\}| \leq |A| + \sum\limits_{i<n} |A_i|$ and, if $B<_\mathcal A A$, then $|B|\leq \max (|A|,|A_0|, \dots, |A_{n-1}|)$.
\end{observation}

\subsection{The system \texorpdfstring{$\elndt$}{eLNDT}}
We now recall the system for NBPs introduced in \cite{DBLP:conf/csl/BussDasKnop20}.
The language of the system $\elndt$ includes just the eNDT formulas. 
A \emph{sequent} is an expression $\Gamma \seqar \Delta$, where $\Gamma$ and $\Delta $ 
are multisets of eNDT formulas (`$\seqar$' is just a syntactic delimiter).
Semantically, such a sequent is interpreted as a judgement ``some formula of $\Gamma$ is false or some formula of $\Delta$ is true''.

Notice that the semantic interpretation of eNDT formulas we gave in \cref{dfn:sat-endt-wrt-extax} means that $\dec A p B$ is logically equivalent to both $(\dual p \cand A) \lor (p \cand B)$ and $(\dual p \cimp A) \cand (p \cimp B) $. 
It is this observation which naturally yields the following system for eNDT sequents from \cite{DBLP:conf/csl/BussDasKnop20}:

\begin{definition}
[Systems $\lndt$ and $\elndt$]
The system $\lndt$ is given by the rules in \cref{fig:lndt}.
An $\lndt$ \emph{derivation} of $\Gamma \seqar                 \Delta$ from \emph{hypotheses} $\mathcal H = \{\Gamma_i \seqar                 \Delta_i\}_{i\in I}$ is defined as expected: it is a finite list of sequents, each either some $\Gamma_i \seqar                 \Delta_i$ from $\mathcal H$ or following from previous ones by rules of $\lndt$, ending with $\Gamma \seqar                 \Delta$.

An $\elndt$ \emph{proof} is just an $\lndt$ derivation from hypotheses that are a set of extension axioms $\mathcal A = \{ e_i \extiff A_i(e_j)_{j<i} \}_{i<n} $; 
here we construe $A \extiff B$ as an abbreviation for the pair of sequents $A \seqar B $ and $B \seqar A$.
We (typically) require that the conclusion of an $\elndt$ proof is free of extension variables.

The \emph{size} $|P|$ of a proof or derivation $P$ is just the number of symbols occurring in it.
\end{definition}

Note that, despite the final condition that conclusions of $\elndt$ proofs are free of extension variables, 
we may sometimes consider intermediate `proofs' with extension variables in the conclusions.
In these cases we will always make explicit the underlying set of extension axioms.

\begin{figure}
    \textbf{Initial sequents and cut:}
    \smallskip
    \[   
    \vlinf{ 0}{}{0 \seqar                 }{}
    \qquad
    \vlinf{ 1}{}{\seqar                 1}{}
    \qquad
    \vlinf{\id}{}{p \seqar                 p}{}
    \qquad
    \vliinf{\cut}{}{\Gamma \seqar                 \Delta}{\Gamma \seqar                 \Delta, A}{\Gamma, A \seqar                 \Delta}
    \]
    
    \medskip
    
    \textbf{Structural rules:}
    \smallskip
    \[
\vlinf{\lefrul \wk}{}{\Gamma, A \seqar                 \Delta}{\Gamma \seqar                 \Delta}
\qquad
\vlinf{\rigrul \wk}{}{\Gamma \seqar                 \Delta, A}{\Gamma \seqar                 \Delta}
\qquad
\vlinf{\lefrul\cntr}{}{\Gamma, A \seqar                 \Delta}{\Gamma, A, A \seqar                 \Delta}
\qquad
\vlinf{\rigrul\cntr}{}{\Gamma \seqar                 \Delta, A}{\Gamma \seqar                 \Delta, A, A}
\]

\medskip

\textbf{Logical rules:}
\smallskip

\[
\vliinf{\lefrul p}{}{\Gamma, \dec A p B \seqar                 \Delta}{\Gamma, A \seqar                 \Delta, p}{\Gamma, p, B \seqar                 \Delta}
\qquad
\vliinf{\rigrul p}{}{\Gamma \seqar                 \Delta, \dec A p B}{\Gamma \seqar                 \Delta, A, p}{\Gamma, p \seqar                 \Delta , B}
\]
\[
\vliinf{\lefrul \lor}{}{\Gamma, A \lor B \seqar                 \Delta}{\Gamma , A \seqar                 \Delta}{\Gamma, B \seqar                 \Delta}
\qquad
\vlinf{\rigrul \lor}{}{\Gamma \seqar                 \Delta, A \lor B}{\Gamma \seqar                 \Delta, A, B}
\]
    \caption{Rules for system $(\mathsf e)\lndt$.}
    \label{fig:lndt}
\end{figure}

\begin{remark}
[Differences from original $\elndt$]
\label{differences-from-original-elndt}
In order to ease the exposition, we have slightly adjusted the definition of $\elndt$ from \cite{DBLP:conf/csl/BussDasKnop20}.
The variations are minor and, in particular, the current presentation is polynomially equivalent to that of \cite{DBLP:conf/csl/BussDasKnop20}.
Nonetheless, let us survey these differences here:
\begin{itemize}
    \item We admit constants $0$ and $1$ within the language. As mentioned in \cite{DBLP:conf/csl/BussDasKnop20}, this does not significantly affect proof size, since $0$ can be encoded as $\dec p p {\dual p}$ and $1$ as $\dec {\dual p} p p $, for an arbitrary propositional variable $p$.
    \item We do not have symbols for negative literals, to facilitate our later definition of `positivity'.
    Note, however, that $\dual p$ is equivalent to the formula $\dec 1 p 0$ in our language.
    \item More generally, we admit decisions on only positive literals, not negative ones, for the same reason. 
    Again, a formula $\dec A {\dual p} B$ may be replaced by the equivalent one $\dec B p A$.
\end{itemize}

\end{remark}

As shown in \cite{DBLP:conf/csl/BussDasKnop20}, the system $\elndt$ is adequate for reasoning about non-determi\-nis\-tic decision trees:

\begin{proposition}
[Soundness and completeness, \cite{DBLP:conf/csl/BussDasKnop20}]
 $\elndt$ proves a sequent $\Gamma \seqar \Delta$ (without extension variables) if and only if $\bigwedge \Gamma \cimp \bigvee \Delta$ is valid.
\end{proposition}
One sanity check here is that the set of valid extension-free sequents is indeed $\coNP$-complete, and so comprises an adequate logic for proof complexity.
This is shown explicitly in \cite{DBLP:conf/csl/BussDasKnop20}, but is also subsumed by the analogous statement for the `positive' fragment of this language that we consider in the next section, namely \cref{prop:val-pos-seqs-conp-complete}.

\section{Monotone functions and positive proofs}
In this section we shall recall monotone Boolean functions and positive NBPs that compute them, and introduce a restriction of the system $\elndt$ that reasons only with such positive NBPs.

\subsection{Monotone Boolean functions and positive programs}

A Boolean function $f: \bool^n \to \bool$ is usually called `monotone' if, whenever $\vec c\in \bool^n$ is obtained from $\vec b \in \bool^n$ by flipping $0$s to $1$s, we have $f(\vec b) \leq f(\vec c)$.
Rephrasing this into our setting formally:
\begin{definition}
[Monotonicity]
Given assignments $\alpha,\beta$, we write $\alpha \leq \beta$ if, for all propositional variables $p$, we have $\alpha(p)\leq \beta(p)$, i.e.\ if $\alpha(p) = 1$ then also $\beta(p)=1$.
A Boolean function $f$ is \emph{monotone} if $\alpha \leq \beta \implies f(\alpha) \leq f(\beta)$.
\end{definition}

There are several known non-uniform `positive' models for computing monotone functions, e.g.\ $\cnot$-free circuits or formulas, monotone span programs \cite{KarWig93:on-span-programs}, and, in our setting, positive NBPs:

\begin{definition}
[Positive NBPs, e.g.\ \cite{GrigniSnipser}]
A NBP is \emph{positive} if, for every $0$-edge from a node $u$ to a node $v$, there is also a $1$-edge from $u$ to $v$.
\end{definition}

\begin{figure}
    \centering
    \begin{tikzpicture}[scale=1, auto,swap]
\foreach \pos/\name/\disp in {
  {(0,4)/1/$p_1$}, 
  {(-1,3)/2/$p_2$},
  {(1,3)/3/$p_2$}, 
  {(-2,2)/4/$p_3$}, 
  {(0,2)/5/$p_3$},
  {(2,2)/6/$p_3$}, 
  {(-3,1)/8/$p_4$},
  {(-1,1)/10/$p_4$},
  {(1,1)/11/$p_4$},
  {(3,1)/12/$p_4$},
  {(-4,-0)/15/$0$},
  {(-2,-0)/16/$0$},
  {(0,-0)/17/$1$},
  {(+2,-0)/18/$0$},
  {(+4,-0.008)/19/$0$}}
\node[minimum size=20pt,inner sep=0pt] (\name) at \pos {\disp};

    \draw [->][thick,dotted](1) to (2);
    \draw [->][thin](1) to (3);
    
    \draw [->][thick,dotted](2) to (4);
    \draw [->][thin](2) to (5);
     
    \draw [->][thick,dotted](3) to (5);
    \draw [->][thin](3) to (6);

     \draw [->][thick,dotted](4) to (8);
     \draw [->][thin](4) to (10);
      \draw [->][thick,dotted](5) to (10);
    \draw [->][thin](5) to (11);
     \draw [->][thick,dotted](6) to (11);
     \draw [->][thin](6) to (12);
     \draw [->][thin](8) to (16);
     \draw [->][thick,dotted](8) to (15);
      \draw [->][thin](10) to (17);
     \draw [->][thick,dotted](10) to (16); \draw [->][thin](11) to (18);
     \draw [->][thick,dotted](11) to (17); \draw [->][thin](12) to (19);
     \draw [->][thick,dotted](12) to (18);
     \draw [->][thin]
    (1) [out=180, in=100] to  (2); \draw [->][thin]
    (2) [out=180, in=100] to  (4); \draw [->][thin]
    (4) [out=180, in=100] to  (8); \draw [->][thin]
    (8) [out=180, in=100] to  (15); \draw [->][thin]
    (3) [out=180, in=100] to  (5); \draw [->][thin]
    (5) [out=180, in=100] to  (10); \draw [->][thin]
    (10) [out=180, in=100] to  (16); \draw [->][thin]
    (6) [out=180, in=100] to  (11); \draw [->][thin]
    (11) [out=180, in=100] to  (17); \draw [->][thin]
    (12) [out=180, in=100] to  (18);
\end{tikzpicture}
    \caption{Positive closure of the OBDD for 2-out-of-4 Exact from \cref{fig:exact-obdd-4-2}.}
    \label{fig:pos-clo-ex-4-2}
\end{figure}

\begin{example}
    The NBP in \cref{fig:pos-clo-ex-4-2} is positive: for each $0$-edge there is a `parallel' $1$-edge.
\end{example}

\begin{proposition}
\label{fact:pos-nbp-comp-mon-fn}
A positive NBP computes a monotone Boolean function.
\end{proposition}
\begin{proof}
[Proof sketch]
Suppose $\alpha\leq \beta$ and $\alpha(G) = 1$. Let $\vec v = (v_0, \dots, v_n)$ be an accepting run, where $v_n = 1$ and, for $i<n$, each $v_i$ is labelled by some propositional variable $p_i$. 
We argue that $\vec v$ is also an accepting run of $\beta$. 
The critical case is when $\alpha (p_i) = 0$ but $\beta(p_i)=1$; in which case the positivity condition on $G$ ensures that there is nonetheless a $1$-edge from $v_i$ to $v_{i+1}$.

\end{proof}

\begin{example}
\label{ex:pos-clo-ex-4-2-computs-thresh-4-2}
    The positive NBP in \cref{fig:pos-clo-ex-4-2} in fact computes the 2-out-of-4 threshold function, returning $1$ if \emph{at least} two of its four inputs are $1$, which is a monotone function.
    This is verifiable by inspection, but we shall soon see that it is a special case of a more general class of positive NBPs.
\end{example}

\subsection{A digression on monotone complexity and closures}
NBPs are a non-uniform version of non-deterministic logspace ($\NL$): each $\NL$ language is accepted by a polynomial-size family of NBPs and, conversely, the evaluation problem for NBPs is complete for $\NL$.
In particular, $\NL$ is precisely the class of languages accepted by, say, $\Logspace$-uniform families of NBPs.

Naturally, our positive NBPs correspond to a positive version of $\NL$ too, called $\mon\NL$ by Grigni and Sipser in \cite{GrigniSnipser,grigni1991structure}.
Those works present a comprehensive development of positive models of computation and their underlying theory.
In particular there is a well-behaved notion of positive non-deterministic Turing machine based on a similar idea to that of positive NBPs: roughly speaking, whenever a transition is available when reading a $0$, the same transition is available when reading a $1$.
It turns out that the class $\mon\NL$, induced by this machine model restricted to logarithmic size work tapes, is equivalent to the class of languages recognised by $\Logspace$-uniform families of positive NBPs \cite{GrigniSnipser}.

One natural construction that is available in the NBP setting (as opposed to, say, Boolean formulas or circuits) is the notion of a `positive closure':

\begin{definition}
[Positive closure of a NBP]
\label{def:pos-clo}
For a NBP $G$ with $0$-edges $E_0$ and $1$-edges $E_1$, we write $\posclo G$ for the NBP with the same vertex set and $0$-edges $E_0$ and $1$-edges $E_0\cup E_1$. I.e., $\posclo G$ is obtained from $G$ by adding, for every $0$-edge from a node $u$ to a node $v$, a $1$-edge from $u$ to $v$ (if there is not already one).
\end{definition}

\begin{example}
\label{ex:pos-clo-ex-4-2}
The positive NBP given in \cref{fig:pos-clo-ex-4-2} is in fact the positive closure of the OBDD for 2-out-of-4 Exact from \cref{fig:exact-obdd-4-2}.
\end{example}

Note that $\posclo G$ is always a positive NBP, by definition.
Thus this construction gives us a `canonical' positive version of a NBP.
In many (but not all) cases, we can precisely characterise the semantic effect of taking positive closures, thanks to the following notion:

\begin{definition}
[Monotone closure]
\label{dfn:mon-clo-bool-fn}
For a Boolean function $f$, we define its \emph{monotone closure} $\monclo f$, by $\monclo f (\alpha) = 1$ if $\exists \beta \leq \alpha . f(\beta) = 1$. 
\end{definition}
The point of the monotone closure of a function $f$ is that it is the `least' monotone function that dominates $f$, i.e.\ such that $f\leq \monclo f$.
There is also a dual notion of the `greatest' monotone function dominated by $f$ which is similarly related to `positive co-NBPs', but we shall not make use of it in this work.

In certain cases, the positive closure of a NBP $G$ computes \emph{precisely} the monotone closure of the function computed by $G$.
Call a NBP \emph{read-once} if, on each path, each propositional variable appears at most once.
We have:

\begin{propC}
[\cite{GrigniSnipser}]
\label{pos-clo-read-once-mon-clo}
Let $G$ be a read-once NBP computing a Boolean function $f$.
Then $\posclo G$ computes $\monclo f$.
\end{propC}
\begin{proof}
[Proof sketch]
Suppose $\monclo f (\alpha) = 1$, and let $\vec v$ be an accepting run of $G$ on some $\beta \leq \alpha$.
Notice that $\vec v$ is also accepting for $\posclo G$ on $\alpha$: at any node $v_i$ labelled by some $p$ on which $\alpha $ and $\beta$ differ, i.e.\ $\beta(p) = 0$ and $\alpha(p) = 1$, by positivity we also have a $1$-edge  $v_i$ to $v_{i+1}$.

Now suppose that $\posclo G$ accepts $\alpha$ by a run $\vec v$ of nodes labelled by $\vec p$ respectively. Define $\beta  \leq \alpha$ by $\beta(p_i) = 1$ if there is no $0$-edge from $v_i$ to $v_{i+1}$ in $G$, otherwise $\beta(p_i) = 0$.
Note that $\beta$ is indeed well-defined, by the read-once property, and indeed $\beta \leq \alpha$ by definition of $\posclo G$.
\end{proof}

In particular, the above result holds when $G$ is `ordered' (an `OBDD'), i.e.\ propositional variables occur in the same relative order in each path through $G$.

\begin{example}
[Monotone closure of Exact]
The monotone closure of the 2-out-of-4 Exact function is the 2-out-of-4 Threshold function, returning 1 if \emph{at least} two of its four inputs are 1.
Since the branching programs from \cref{fig:exact-obdd-4-2} were read-once, by inspection, the fact that \cref{fig:pos-clo-ex-4-2} computes the 2-out-of-4 Threshold function is justified by \cref{pos-clo-read-once-mon-clo}, cf.~\cref{ex:pos-clo-ex-4-2,ex:pos-clo-ex-4-2-computs-thresh-4-2} earlier.
\end{example}

Note, however, that the result does not hold for arbitrary NBPs.

In fact, there is no feasible notion of `positive closure' on NBPs that always computes the monotone closure, due to the following result:
\begin{thmC}
[\cite{GrigniSnipser}]
There are monotone functions computed by polynomial-size families of NBPs, but no polynomial-size family of positive NBPs.
\end{thmC}

Note that this result is the analogue for the NBP model to Razborov's seminal results for the circuit model \cite{razborov1985,Raz85:lower-bds-mon-comp-log-permanent}. 
This result follows by establishing a non-uniform version of `$ \co\mon\NL \not\subseteq \mon\NL $'; in particular there is a monotone $\co\NL$ language (namely non-reachability in a graph) computed by no polynomial-size family of positive NBPs. 
The result above now follows by
 the Immerman-Szelepcs\'enyi theorem that $\NL = \co\NL$ \cite{Imm88:nspace-closed-under-complement,Sze88:method-forced-enumeration-nd-automata} and $\NL$-completeness of NBP evaluation.

\subsection{Representations of positive branching programs}
Let us now return to our representation of NBPs by extended formulas.
Recall that we implement non-determinism using disjunction,
so we may duly define the corresponding notion of positivity at the level of eNDT formulas themselves:

\begin{definition}
[Positive formulas]
An eNDT formula is \emph{positive} if, for each subformula of the form $\dec A p B$, we have $B = A \lor C$ for some $C$.

A set of extension axioms $\mathcal A = \{e_i \extiff A_i\}_{i<n}$ is \emph{positive} if each $A_i$ is positive.
\end{definition}

Positive eNDT formulas, under positive extension axioms, are just representations of positive NBPs.

\begin{example}
[2-out-of-4 Threshold, revisited]
Recall the positive NBP from \cref{fig:pos-clo-ex-4-2}.
We can represent this by the positive formula $\ee 1 1$ under the positive set of extension axioms:
    \[
\begin{array}{r@{\ \extiff \ }l}
    \ee 1 1 & \posdec {\ee 2 1 } {p_1}{\ee 2 2} \\
    \ee 2 1 & \posdec {\ee 3 1 } {p_2}{\ee 3 2} \\
    \ee 2 2  & \posdec {\ee 3 2 } {p_2}{\ee 3 3}
\end{array}
\qquad
\begin{array}{r@{\ \extiff \ }l}
    \ee 3 1  & \posdec {\ee 4 1 } {p_3}{\ee 4 2} \\
    \ee 3 2  & \posdec {\ee 4 2 } {p_3}{\ee 4 3} \\
    \ee 3 3  & \posdec {\ee 4 3 } {p_3}{\ee 4 4}
\end{array}
\qquad
\begin{array}{r@{\ \extiff \ }l}
    \ee 4 1  & \posdec 0 {p_4} 0 \\
    \ee 4 2  & \posdec 0 {p_4} 1 \\
    \ee 4 3  & \posdec 1 {p_4} 0 \\
    \ee 4 4  & \posdec 0 {p_4} 0
\end{array}
\]
Each $\ee i j $ represents the $j$\textsuperscript{th} node (left to right) on the $i$\textsuperscript{th} row (top to bottom), just like in \cref{ex:exact24-ext-axioms}. 
Indeed the positive set of extension axioms above may be seen formally as the positive closure of \eqref{eq:exact24-ext-axioms} defined at the level of eNDT formulas and extension axioms, cf.~\cref{def:pos-clo}.
\end{example}

Notice in particular that a positive decision $\posdec A p B$ is semantically equivalent to $A \lor (p \cand B)$, which is monotone in $A$, $p$ and $B$.
Since every other symbol/connective also computes a monotone function we may thus directly obtain the analogue of \cref{fact:pos-nbp-comp-mon-fn}:

\begin{proposition}
Suppose $\mathcal A = \{e_i \extiff A_i\}_{i<n}$ is a set of positive extension axioms. 
Each positive eNDT formula $A$ over $e_0, \dots, e_{n-1}$ computes a monotone Boolean function with respect to $\mathcal A$.
\end{proposition}
This argument proceeds by $\mathcal A$-induction on $A$ and is routine.
Other than the fact that all connectives are monotone, we also rely on the fact that we have no negative literals in our language.

\subsection{A system for positive branching programs}

The semantic equivalence of $\posdec A p B$ and $A \lor (p \cand B)$ motivates the following `positive decision' rules:

\begin{equation}
    \label{eq:pos-dec-rules}
    \vliinf{\lefrul{\pos p}}{}{\Gamma, \posdec A p B \seqar\Delta}{\Gamma, A \seqar                 \Delta}{\Gamma, p, B \seqar\Delta}
\qquad
\vliinf{\rigrul{\pos p}}{}{\Gamma \seqar                 \Delta, \posdec A p B}{\Gamma \seqar                 \Delta, A,p}{\Gamma \seqar\Delta, A, B}
\end{equation}
Naturally, these rules satisfy the subformula property and, as well as being sound, are moreover \emph{invertible}: the validity of the conlusion implies the validity of each premiss. 
This will be important for our completeness argument shortly in \cref{prop:poselndt-sound-compl}.

As expected, none of the arguments $A$, $p$ or $B$ above `change sides' by the rules above.
This is due to the fact that positive decisions are, indeed, monotone, unlike general decisions, for which the decision variable $p$ may change sides.
Note, however, that the two rules above are not `dual', in the logical sense, unlike general decisions.
This is because positive decisions are no longer self-dual.

Notably, the two positive decision rules comprise reasoning already available in $\elndt$: 

\begin{example}
    [Deriving positive rules]
    The two rules in \eqref{eq:pos-dec-rules} may be derived in $\elndt$ in constantly many lines:
    \[
    \vlderivation{
    \vliin{\lefrul p }{}{\Gamma , \posdec A p B \seqar \Delta }{
        \vlin{\rigrul\wk}{}{\Gamma, A \seqar \Delta,p}{
        \vlhy{\Gamma, A \seqar \Delta}
        }
    }{
        \vliin{\lefrul \lor}{}{\Gamma, p , A\lor B\seqar \Delta}{
            \vlin{\lefrul\wk}{}{\Gamma, p , A \seqar \Delta}{
            \vlhy{\Gamma, A \seqar \Delta}
            }
        }{
            \vlhy{\Gamma, p, B \seqar \Delta}
        }
    }
    }
    \qquad
    \vlderivation{
    \vliin{\rigrul p}{}{\Gamma \seqar \Delta, \posdec A p B}{
        \vlhy{\Gamma \seqar \Delta, A, p}
    }{
        \vlin{\lefrul\wk}{}{\Gamma, p \seqar \Delta, A\lor B}{
        \vlin{\rigrul \lor}{}{\Gamma  \seqar \Delta, A\lor B}{
        \vlhy{\Gamma \seqar \Delta, A, B}
        }
        }
    }
    }
    \]
\end{example}

This motivates the definition of our `positive subsystem' of $\elndt$:

\begin{definition}
[System $\poselndt$]
The system $\poselndt$ is defined just like $\elndt$, except replacing the $\lefrul p$ and $\rigrul p$ rules by the positive ones above in Equation~\eqref{eq:pos-dec-rules}.
Moreover, all extension axioms and formulas occurring in a proof (in particular cut-formulas) must be positive.
\end{definition}

As for $\elndt$, we should verify that the set of valid positive sequents (without extension variables) is actually sufficiently expressive to be meaningful for proof complexity, i.e.\ that they are $\coNP$-complete.
While this is fairly immediate for other positive systems, such as $\mlk$, it is not so clear here so we give a self-contained argument.

\begin{proposition}
\label{prop:val-pos-seqs-conp-complete}
 The set of valid positive sequents (without extension variables) is $\coNP$-complete.
\end{proposition}
\begin{proof}
 By the Cook-Levin theorem \cite{Coo71:comp-thm-prv-procs,Lev73:univ-seq-search-problems}, we know that the validity problem for DNFs is $\coNP$-complete, so we will show how to encode a DNF as an equi-valid positive sequent. 
 
 First, note that we may express positive terms (i.e.\ conjunctions of propositional variables) as positive NDT formulas by exploiting the equivalence $\posdec 0 p B \iff 0 \lor (p \land B) \iff p \land B$.
 Recursively applying this equivalence we obtain:
 \begin{equation}
     \label{eq:posterm-to-posndt}
     \bigwedge_{i=1}^m p_i
 \iff
 \posdec 0 {p_1} {\posdec 0 {p_2} { (\cdots \posdec 0 {p_{m-1}} {p_m )\cdots} } }
 \end{equation}
 Let us write $\posterm{p_1, \dots, p_m}$ for the positive NDT formula on the right above.
 
 Now, fix a DNF instance $A$ over propositional variables $p_1, \dots , p_k$ and let $A'$ be obtained by replacing each negative literal $\dual p_i$ by a fresh (positive) propositional variable $p_i'$.
 Write $A' = \bigvee\limits_{i=1}^n \bigwedge \vec p_i$, where each $\vec p_i$ is a sequence of propositional variables (among $p_1 , \dots, p_k, p_1' , \dots, p_k'$).
 Now we have that the following positive NDT sequent is equi-valid with $A$, as required:
 \[
 p_1 \lor p_1', \dots, p_k \lor p_k'
 \seqar
 \posterm{\vec p_1}, \dots, \posterm{\vec p_n} \qedhere
 \]
\end{proof}

Finally, our calculus $\poselndt$ is indeed adequate for reasoning about positive sequents:

\begin{proposition}
 [Soundness and completeness]
 \label{prop:poselndt-sound-compl}
 $\poselndt$ proves a positive sequent $\Gamma \seqar \Delta$ (without extension variables) if and only if $\bigwedge \Gamma \cimp \bigvee \Delta$. 
\end{proposition}
\begin{proof}
[Proof sketch]
Similarly to the argument in [BDK20] for $\elndt$, we may proceed by cut-free proof search, and will not make use of any extension variables.
Note that each logical rule is \emph{invertible}, i.e.\ the validity of the conclusion implies the validity of each premiss.
Moreover, each premiss (of a logical rule) has fewer connectives than the conclusion.
Thus, bottom-up, we may simply repeatedly apply the logical steps until we reach sequents of only atomic formulae. 
Such a sequent is valid if and only if there is a $0$ on the LHS, a $1$ on the RHS, or some propositional variable $p$ on both sides.
Each of these cases may be derived from initial sequents using the weakening rules, $\lefrul \wk$ and $\rigrul \wk$.
\end{proof}

\begin{remark}
[Completeness with respect to extension axioms]
\label{rem:compl-wrt-extaxs}
While it is standard to only consider extended proofs over extension-free theorems, let us point out that we also have a stronger version of completeness with respect to sets of (positive) extension axioms.

Given a set of (positive) extension axioms $\mathcal A = \{ e_i \extiff A_i\}_{i<n} $, a (positive) formula $A$ over $e_0, \dots, e_{n-1}$ is \emph{$\mathcal A$-valid} if, for every assignment $\alpha$, we have $\alpha \sat{\mathcal A} A$ (cf.~\cref{dfn:sat-endt-wrt-extax}).

The same argument as in Proposition~\ref{prop:poselndt-sound-compl} can now be applied to show completeness for $\mathcal A$-valid positive sequents by proofs using the extension axioms $\mathcal A$.
The only difference is that, when we reach a connective-free sequent (bottom-up), extension variables may also occur, not just propositional variables and constants.
In this case we must use the extension axioms to unwind the extension variables and continue the proof search algorithm.
Termination of this process now follows by appealing to $\mathcal A$-induction (cf.~\cref{A-induction}).
\end{remark}

\subsection{Some basic theorems}
Let us now present some basic theorems of $\poselndt$, which will all have polynomial-size proofs.
These will be useful for our later arguments and,
at the same time, exemplify how we will conduct proof complexity theoretic reasoning in what follows.

Let us first point out an expected property, that we can polynomially derive a general identity rule from the atomic version included in the definition of $\poselndt$.
Albeit a simple observation, it has the consequence that applying substitutions of formulas for variables in proofs has only  polynomial overhead in proof size.
\begin{proposition}
[General identity]
\label{prop:gen-identity}
Let $\mathcal A = \{ e_i \extiff A_i \}_{i<n}$ be a set of positive extension axioms.
There are polynomial-size $\poselndt$ proofs of $A \seqar                 A$, for positive formulas $A$ containing only extension variables among $e_0, \dots, e_{n-1}$.
\end{proposition}
\begin{proof}
 We construct a (dag-like) proof of the required sequent by $\mathcal A$-induction.
 More precisely, for each such $A$, we construct a polynomial-size proof containing sequents $B \seqar                 B$ for each $B\leq_\mathcal A A$ by $\mathcal A$-induction on $A$.
 \begin{itemize}
     \item If $A$ is a propositional variable then we are done by the rule $\id$. 
     \item If $A = 0$ then we have:
     \[
     \vlderivation{
     \vlin{\rigrul\wk}{}{0 \seqar                 0}{
     \vlin{0}{}{0 \seqar                 }{\vlhy{}}
     }
     }
     \]
     \item If $A=1$ then we have:
     \[
     \vlderivation{
     \vlin{\lefrul \wk}{}{1 \seqar                 1}{
     \vlin{1}{}{\seqar                 1}{\vlhy{}}
     }
     }
     \]
     \item If $A = e_i$ for some $i<n$ and $e_i \dseqar A_i$ is an extension axiom of $\mathcal A$, then we simply cut the two extension axiom sequents against each other: 
     \[
     \vliinf{\cut}{}{e_i \seqar e_i}{e_i \seqar A_i}{A_i \seqar e_i}
     \]

     \item If $A = B \lor C$ then we extend the proof obtained by the inductive hypothesis as follows,
     \[
     \vlderivation{
     \vlin{\rigrul \lor}{}{B \lor C \seqar                 B \lor C}{
     \vliin{\lefrul \lor}{}{B \lor C \seqar                 B,C}{
        \vlin{\rigrul{\wk}}{}{B \seqar                 B,C}{
        \vliq{\IH}{}{B \seqar                 B}{\vlhy{}}
        }
     }{
        \vlin{\rigrul \wk}{}{C \seqar                 B,C}{
        \vliq{\IH}{}{C \seqar                 C}{\vlhy{}}
        }
     }
     }
     }
     \]
     where sequents marked $\IH$ are obtained by the inductive hypothesis.
     \item If $A = \posdec B p C$ then we extend the proof obtained by the inductive hypothesis as follows:
     \[
     \vlderivation{
        \vliin{\lefrul{\pos p}}{}{\posdec B p C \seqar                 \posdec B p C}{
            \vliin{\rigrul{\pos p}}{}{B \seqar                 \posdec B p C}{
                \vlin{\rigrul \wk}{}{B \seqar                 B,p}{
                \vliq{\IH}{}{B \seqar                 B}{\vlhy{}}
                }
            }{
                \vlin{\rigrul \wk}{}{B \seqar                 B,C}{
                \vliq{\IH}{}{B \seqar                 B}{\vlhy{}}
                }
            }
        }{
            \vliin{\rigrul{\pos p}}{}{p,C \seqar                 \posdec B p C}{
                \vlin{\lefrul \wk , \rigrul \wk}{}{p,C \seqar                 B,p}{
                \vlin{\id}{}{p \seqar                 p}{\vlhy{}}
                }
            }{
                \vlin{\lefrul \wk, \rigrul \wk}{}{p,C \seqar                 B,C}{
                \vliq{\IH}{}{C \seqar                 C}{\vlhy{}}
                }
            }
        }
     }
     \]

    where sequents marked $\IH$ are obtained by the inductive hypothesis. 
 \end{itemize}
 
 To evaluate proof size note that, at each step of the argument above, we add a constant number of lines of polynomial size in $A$ and $\mathcal A$.
 Thus a polynomial bound follows by \cref{complexity-of-A-induction}.
\end{proof}

Notice, in the final step above, that we do not formally `duplicate' the subproof for $B\seqar B$ as this, recursively applied, could cause an exponential blowup.
This is why the construction by $\mathcal A$-induction is phrased as constructing a \emph{single} proof that contains all `smaller' instances of identity already, with inductive steps just extending that proof.
In what follows we shall be less rigorous when constructing formal proofs in this way,
simply saying that we `construct them by $\mathcal A$-induction'.
We shall also typically leave proof complexity analysis like the one above implicit.

For our later simulations, the following `truth conditions' for positive decisions will prove useful:
\begin{proposition}
 [Truth conditions]
\label{Truth}
Let $\mathcal A=\{e_i \extiff A_i\}_{i<n}$ be a set of positive extension axioms and let $A$ and $B$ be formulas over $e_0, \dots, e_{n-1}$.
There are polynomial-size $\poselndt$ proofs of the following sequents with respect to $\mathcal A$:
\begin{enumerate}
    \item\label{item:truth-posdec-implies-0case} $\posdec A p B \seqar                 A,p$
    \item\label{item:truth-posdec-implies-1case} $\posdec A p B \seqar                 A,B$
    \item\label{item:truth-0case-implies-posdec} $A \seqar                 \posdec A p B$
    \item\label{item:truth-1case-implies-posdec} $p,B \seqar                 \posdec A p B$
\end{enumerate}

\end{proposition}

\begin{proof}
We give the proofs explicitly:
\[
\cref{item:truth-posdec-implies-0case} : \ \vlderivation{
    \vliin{\lefrul {\pos p} }{}{\posdec A p B \seqar                 A, p}{
        \vlin{\rigrul \wk}{}{A \seqar                 A,p}{
        \vlin{\id}{}{A \seqar                 A}{\vlhy{}}
        }
    }{
        \vlin{\lefrul \wk, \rigrul \wk}{}{p,B \seqar                 A,p}{
        \vlin{\id}{}{p \seqar                 p}{\vlhy{}}
        }
    }
}
\qquad
\cref{item:truth-posdec-implies-1case} : \
\vlderivation{
    \vliin{\lefrul {\pos p} }{}{\posdec A p B \seqar                 A, B}{
        \vlin{\rigrul \wk}{}{A \seqar                 A,B}{
        \vlin{\id}{}{A \seqar                 A}{\vlhy{}}
        }
    }{
        \vlin{\lefrul \wk, \rigrul \wk}{}{p,B \seqar                 A,B}{
        \vlin{\id}{}{B \seqar                 B}{\vlhy{}}
        }
    }
}
\]
\[
\cref{item:truth-0case-implies-posdec} : \ 
\vlderivation{
    \vliin{\rigrul{\pos p}}{}{A \seqar                 \posdec A p B}{
        \vlin{\rigrul \wk}{}{A \seqar                 A,p}{
        \vlin{\id}{}{A \seqar                 A}{\vlhy{}}
        }
    }{
        \vlin{\rigrul \wk}{}{A \seqar                 A,B}{
        \vlin{\id}{}{A \seqar                 A}{\vlhy{}}
        }
    }
}
\qquad
\cref{item:truth-1case-implies-posdec} : \ 
\vlderivation{
    \vliin{\rigrul{\pos p}}{}{p,B \seqar                 \posdec A p B}{
        \vlin{\lefrul \wk, \rigrul \wk}{}{p,B \seqar                 A,p}{
        \vlin{\id}{}{p \seqar                 p}{\vlhy{}}
        }
    }{
        \vlin{\lefrul \wk, \rigrul \wk}{}{p,B \seqar                 A,B}{
        \vlin{\id}{}{B \seqar                 B}{\vlhy{}}
        }
    }
}
\]
where the steps marked $\id$ are derivable by Proposition~\ref{prop:gen-identity}.
\end{proof}

Notice that, given that we have polynomial-size proofs for general identity, \cref{prop:gen-identity}, the result above also just follows immediately from semantic validity of the sequents \cref{item:truth-posdec-implies-0case}-\cref{item:truth-1case-implies-posdec} and completeness of $\poselndt$, \cref{prop:poselndt-sound-compl}, by simply substituting the formulas $A$ and $B$ for appropriate constant-size instances of \cref{item:truth-posdec-implies-0case}-\cref{item:truth-1case-implies-posdec}.
We gave the argument explicitly to exemplify formal proofs of the system $\poselndt$.
We shall, however, make use of the aforementioned observation in the remainder of this work.

\begin{example}
[A positive `medial']
\label{pos-medial-example}
\label{pos-medial}
Branching programs enjoy elegant symmetries.
For instance, in our eNDT notation, we have validity of the following pair of sequents,
\[
\dec {(\dec A q B)} p {(\dec C q D)}
\ \dseqar \ 
\dec {(\dec A p C)} q {(\dec B p D )}
\]
corresponding to a certain permutations of nodes in NBPs.

We also have a \emph{positive} version of the law above, namely:
\begin{equation}
    \label{eq:pos-medial}
    \begin{array}{rl}
      & \posdec {(\posdec A q B)} p {(\posdec C q D)} \\
      \dseqar   & \posdec {(\posdec A p C)} q {(\posdec B p D )} 
    \end{array}
\end{equation}
        The validity of this equivalence can be seen by noticing that each side is equivalent to $A \lor  (p \cand C)  \lor (q \cand B) \lor (p \cand q \cand D) $.
By completeness and substitution, we thus have polynomial-size proofs of \eqref{eq:pos-medial}.

\end{example}

Finally, we will make use of the following consequence of the truth conditions:
\begin{corollary}
[Deep inference]
\label{refined-replacement}
 \label{mon-replacing-result}
There are polynomial size $\poselndt$ derivations of 
\[
\Gamma , \posdec A p B \seqar \Delta, \posdec{A' }p{B'}
\]
from hypotheses $\Gamma, A \seqar \Delta, A'$ and $\Gamma , B \seqar \Delta, B'$, over any positive extension axioms including all extension variables occurring in $A$ and $B$. 
\end{corollary}
\begin{proof}
\renewcommand{\storageone}{\cref{Truth}.\cref{item:truth-0case-implies-posdec}}
\renewcommand{\storagetwo}{\cref{Truth}.\cref{item:truth-1case-implies-posdec}}
We give the derivation below:
\[
\vlderivation{
    \vliin{\lefrul{\pos p}}{}{\Gamma, \posdec A p B \seqar \Delta, \posdec {A'}p{B'} }{
        \vliin{\cut}{}{ \Gamma, A \seqar \Delta, \posdec{A'}{p}{B'} }{
            \vlhy{\Gamma, A \seqar \Delta , A' }
        }{
            \vliq{}{}{A' \seqar \posdec{A'}{p}{B'} }{\vlhy{ \text{\storageone} }}
        }
    }{
        \vliin{\cut}{}{\Gamma, p, B \seqar \Delta, \posdec{A'}{p}{B'} }{
            \vlhy{\Gamma, B \seqar \Delta, B'}
        }{
            \vliq{}{}{p,B' \seqar \posdec{A'}{p}{B'}}{\vlhy{ \text{\storagetwo} }}
        }
    }
}
\]
Note that we have omitted several structural steps, namely $\lefrul \wk, \rigrul \wk$ above $\cut$-steps, to match contexts.
We will typically continue to omit these, freely using `context splitting' and `context sharing' behaviour, under structural rules.
\end{proof}

\section{Programs for counting and their basic properties}
\label{sec:programs-for-counting-and-basic-props}
Let us now consider some of the Boolean counting functions that appeared in our earlier examples more formally.

The \emph{Exact} functions $\exact n k : \bool^n \to \bool$ are defined by:
\[
\exact n k (b_1, \dots, b_n) =1
\quad \iff \quad
\sum\limits_{i=1}^n b_i \, = \,  k
\]
I.e.\
$\exact n k (b_1, \dots, b_n) =1$ iff \emph{exactly} $k$ of $b_1, \dots, b_n$ are $1$.

Taking the monotone closures of these functions (cf.~Definition~\ref{dfn:mon-clo-bool-fn}), we obtain the \emph{Threshold} functions $\thresh n k : \bool^n \to \bool$ by:
\[
\thresh n k (b_1, \dots, b_n) =1
\quad \iff \quad
\sum\limits_{i=1}^n b_i \, \geq  \,  k \geq 0
\]
I.e.\
$\thresh n k (b_1, \dots, b_n) =1$ iff \emph{at least} $k$ of $b_1, \dots, b_n$ are $1$.\footnote{Note that this description of the monotone closure of Exact holds only when $k\geq 0$, hence the side condition in display.}

For consistency with the exposition so far, given a list $\vec p = p_0, \dots, p_{n-1}$ of propositional variables, we construe $\exact n k (\vec p) $ as a Boolean function from assignments to Booleans, writing, say, $\exact n k (\vec p) (\alpha)$ for its (Boolean) output.
Similarly for $\thresh n k (\vec p)$.

\subsection{OBDDs for Exact and their representations}
It is well-known that counting functions like those above are computable by `ordered' branching programs, or `OBDDs' (see, e.g., \cite{Weg00:bps-and-bdds}). 
These are deterministic branching programs where variables occur in the same relative order on each path. 
For instance we give an OBDD for $\exact n k (p_1, \dots, p_n)$ in \cref{fig:exact-obdd}.
\begin{figure}[t]
    \[
\begin{tikzpicture}[scale=1.2, auto,swap]
\foreach \pos/\name/\disp in {
  {(0,4)/1/$p_1$}, 
  {(-1,3)/2/$p_2$},
  {(1,3)/3/$p_2$}, 
  {(-2,2)/4/$p_3$}, 
  {(0,2)/5/$p_3$},
  {(2,2)/6/$p_3$}, 
  {(0,1.53)/7/$\vdots$}, 
  {(-2,1.53)/8/\vdots}, 
  {(2,1.53)/9/\vdots},
  {(-2.4,1.48)/10/},
  {(2.4,1.48)/11/},
  {(-3.4,0.22)/12/$p_n$},
  {(3.4,0.22)/13/$p_n$},
  {(0,0.22)/14/$p_n$},
  {(-4.4,-0.82)/15/$0$},
  {(-2.4,-0.82)/16/$0$},
  {(-1,-0.82)/17/$1$},
  {(+1,-0.82)/18/$0$},
  {(+2.4,-0.82)/19/$0$},
  {(+4.4,-0.82)/20/$0$},
  {(1,1.22)/21/},
  {(-1,1.22)/22/}}
\node[minimum size=20pt,inner sep=0pt] (\name) at \pos {\disp};

    \draw [->][thick,dotted](1) to (2);
    \draw [->][thin](1) to (3);
    
    \draw [->][thick,dotted](2) to (4);
    \draw [->][thin](2) to (5);
     
    \draw [->][thick,dotted](3) to (5);
    \draw [->][thin](3) to (6);
    \draw [-][thick,dotted](4) to (10);
    
    \draw [->][thick,dotted](10) to (12);
    \draw [->][thin](11) to (13);
     \draw [->][thick,dotted](12) to (15);
     \draw [->][thin](12) to (16);
      \draw [->][thick,dotted](21) to (14);
    \draw [->][thin](22) to (14);
     \draw [->][thick,dotted](14) to (17);
     \draw [->][thin](14) to (18);
     \draw [->][thin](13) to (19);
     \draw [->][thick,dotted](13) to (20);
     
\end{tikzpicture}
\]
    \caption{An `ordered' branching program (OBDD) for $\exact n k (p_1, \dots, p_n)$, where there are $k$ $0$s to the left of the $1$, and $n-k$ to the right.}
    \label{fig:exact-obdd}
\end{figure}

Recall that $\lor $ is used in the representation of NBPs only to model nondeterminism. 
Thus the Exact functions can be represented without using disjunction as follows:
\begin{definition}
[Representation of Exact]
\label{def:exact-ext-axs}
For each list $\vec p$ of propositional variables, and each integer $k$, we introduce an extension variable $\ex {\vec p} k$ and write $\exextaxs{}{}$ for the set of all extension axioms of the form (i.e.\ for all choices of $p$, $\vec p$ and $k$),
\begin{equation}
    \label{eq:exact-ext-axs}
    \begin{array}{c@{\quad \extiff \quad}ll}
     \ex \emptylist 0 & 1 & \\
     \ex \emptylist k & 0 & \text{if $k \neq 0$} \\     \ex {p\vec p} k & \dec {\ex {\vec p} k} p {\ex {\vec p} {k-1}}
\end{array}
\end{equation}
where we write $\emptylist$ for the empty list.
By convention we allow $k$ to be negative, for uniformity of the definition.

\end{definition}

\begin{remark}
    [Well-foundedness and complexity of extension axioms]
    \label{rem:infinite-set-of-ext-axs}
Note that, even though $\exextaxs{}{}$ is an infinite set, we may use it as the underlying set of extension axioms for proofs, with the understanding that only finitely many will actually ever be used in a particular proof.
We will typically not explicitly compute this set, but such a consideration will be subsumed by our analysis of proof complexity.

    While the extension variables above (and their axioms) do not strictly follow the subscripting conditions from \cref{extension-axioms-definition}, we may understand them to be `names' for the appropriate subscripting. 
It suffices to establish the well-foundedness of the extension axiom set in \eqref{eq:exact-ext-axs}, which is clear by induction on the length of the superscript.

Note that, strictly speaking, the indices of extension variables also contribute to proof size. 
However note that, if a proof involves only $n$ extension variables, then only indices of length $O(\log n)$ are required, possibly under suitable reindexing. 
Thus, when measuring proof size, we may safely count only the number of propositional and extension variables occurring.
\end{remark}

\begin{proposition}
\label{prop:exact-ext-vars-compute-exact}
Let $\vec p = (p_1, \dots, p_n)$. $\ex {\vec p} k$ computes $\exact n k (\vec p)$, with respect to $\exextaxs{}{}$.
\end{proposition}
\begin{proof}
We show that $\alpha \sat{\mathcal E} \ex {\vec p} k \iff \exact n k (\vec p) (\alpha) = 1$ by induction on the length $n$ of the list $\vec p$. 

 If $n=0$ then $\exact 0 k (\emptylist )$ attains the value $ 1$ iff $k=0$, so the result is immediate from the first two axioms of \eqref{eq:exact-ext-axs}.
For the inductive step we have:
    \[
    \begin{array}[b]{rc@{\ \iff \ \ }ll}
        \alpha \sat{\mathcal E} \ex{p\vec p}{k} & & \alpha \sat{\mathcal E} \dec {\ex {\vec p } k} p {\ex {\vec p} {k-1}} & \text{by \eqref{eq:exact-ext-axs} and \cref{dfn:sat-endt-wrt-extax}}\\
        \noalign{\smallskip}
         & & \begin{cases}
            \alpha \sat{\mathcal E} \ex {\vec p} k & \alpha(p) = 0 \\
            \alpha \sat{\mathcal E} \ex {\vec p} {k-1} & \alpha(p) = 1
         \end{cases} & \\
         \noalign{\smallskip}
        & & \begin{cases}
            \exact n k (\vec p) (\alpha) = 1 & \alpha(p)=0 \\
            \exact n {k-1} (\vec p) (\alpha) = 1 & \alpha(p) = 1
        \end{cases} & \text{by inductive hypothesis}\\
        \noalign{\smallskip}
        & & \exact {n+1} k (p,\vec p) (\alpha) = 1 &
    \end{array}
      \tag*{\qedhere}
    \]
\end{proof}

\subsection{Positive NBPs for Threshold via positive closure}

Notice that, since the Exact programs we considered were OBDDs, which are in particular read-once, the semantic characterisation of positive closure by monotone closure from \cref{pos-clo-read-once-mon-clo} applies.
Looking back to Figure~\ref{fig:exact-obdd}, the positive closures we are after (as NBPs) are given in Figure~\ref{fig:thresh-nbps-pos-clo-of-exact}.
Realising this directly as eNDT formulas and extension axioms we obtain the following:
\begin{figure}
    \[
\begin{tikzpicture}[scale=1.2, auto,swap]
\foreach \pos/\name/\disp in {
  {(0,4)/1/$p_1$}, 
  {(-1,3)/2/$p_2$},
  {(1,3)/3/$p_2$}, 
  {(-2,2)/4/$p_3$}, 
  {(0,2)/5/$p_3$},
  {(2,2)/6/$p_3$}, 
  {(0,1.53)/7/$\vdots$}, 
  {(-2,1.53)/8/\vdots}, 
  {(2,1.53)/9/\vdots},
  {(-2.4,1.48)/10/},
  {(2.4,1.48)/11/},
  {(-3.4,0.22)/12/$p_n$},
  {(3.4,0.22)/13/$p_n$},
  {(0,0.22)/14/$p_n$},
  {(-4.4,-0.82)/15/$0$},
  {(-2.4,-0.82)/16/$0$},
  {(-1,-0.82)/17/$1$},
  {(+1,-0.82)/18/$0$},
  {(+2.4,-0.82)/19/$0$},
  {(+4.4,-0.82)/20/$0$},
  {(-2,0.22)/21/$\hdots$},
  {(+2,0.22)/22/$\hdots$},
  {(1,1.22)/23/},
  {(-1,1.22)/24/}}
\node[minimum size=20pt,inner sep=0pt] (\name) at \pos {\disp};

      \draw [->][thick,dotted](23) to (14);
    \draw [->][thin](24) to (14);
    \draw [->][thin]
    (23) [out=180, in=100] to  (14);
  \draw [->][thin]
    (1) [out=180, in=100] to  (2);
    \draw [->][thick,dotted](1) to (2);
    \draw [->][thin](1) to (3);
    \draw [->][thin]
    (2) [out=180, in=100] to  (4);
    \draw [->][thick,dotted](2) to (4);
    \draw [->][thin](2) to (5);
     \draw [->][thin]
    (3) [out=180, in=100] to  (5);
    \draw [->][thick,dotted](3) to (5);
    \draw [->][thin](3) to (6);
    \draw [-][thick,dotted](4) to (10);
    \draw [->][thin]
    (10) [out=180, in=100] to  (12);
    \draw [->][thick,dotted](10) to (12);
    \draw [->][thin](11) to (13);
    \draw [->][thin]
    (12) [out=180, in=100] to  (15);
     \draw [->][thick,dotted](12) to (15);
     \draw [->][thin](12) to (16);
     \draw [->][thin]
    (14) [out=180, in=100] to  (17);
     \draw [->][thick,dotted](14) to (17);
     \draw [->][thin](14) to (18);
     \draw [->][thin](13) to (20);
     \draw [->][thin]
    (13) [out=180, in=100] to  (19);
     \draw [->][thick,dotted](13) to (19)
     ;
     
\end{tikzpicture}
\]
    \caption{The positive closure of the OBDD for Exact from Figure~\ref{fig:exact-obdd}, computing $\thresh n k (p_1, \dots, p_n)$.
    Again, there are $k$ $0$s to the left of the $1$, and $n-k$ to the right.}    \label{fig:thresh-nbps-pos-clo-of-exact}
\end{figure}

\begin{definition}
[Positive eNDTs for Threshold]
\label{Defin.thresh.}
For each list $\vec p$ of propositional variables, and each integer $k$, we introduce an extension variable $\thr {\vec p} k$ and write $\thrextaxs {}{}$ for the set of all extension axioms of the form (i.e.\ for all choices of $p$, $\vec p$ and $k$):
\begin{equation}
    \label{eq:thresh-ext-axs}
    \begin{array}{c@{\quad \extiff \quad}ll}
     \thr \emptylist 0 & 1 & \\
     \thr \emptylist k & 0 & \text{if $k \neq 0$} \\     \thr{p\vec p} k & \posdec {\thr {\vec p} k} p  {\thr {\vec p} {k-1} }
\end{array}
\end{equation}

\end{definition}

Note above that we allow $k$ to be negative, to be consistent with \cref{def:exact-ext-axs}.

Again, even though $\thrextaxs{}{}$ is an infinite set, we shall typically write $\poselndt$ proofs with respect to this set of extension axioms, with the understanding that only finitely many are ever used in any particular proof (see \cref{rem:infinite-set-of-ext-axs}).

Note that the extension variables $\thr {\vec p} k$ and extension axioms $\thrextaxs{}{}$ above are just the positive closures of $\ex {\vec p } k $ and $\exextaxs{}{}$ earlier, within the eNDT setting.

Thus, under \cref{pos-clo-read-once-mon-clo}, we have from \cref{prop:exact-ext-vars-compute-exact} that, for each non-negative $k$, $\thr {\vec p} k$ computes exactly the threshold function $\thresh n k (\vec p) $ with respect to $\thrextaxs{}{}$:

\begin{corollary}{\label{Thr}}
If $k\geq 0$, then
$\thr {\vec p} k$ computes $\thresh n k (\vec p)$, with respect to $\thrextaxs{}{}$.
\end{corollary}

Note that, for $k$ negative, we could have alternatively set $\thr \epsilon k $ to be $1$.
We could have also simply set $\thr {\vec p} 0 $ to be $1$ for arbitrary $\vec p$.
Instead, we have chosen to systematically take the positive closure of the aforementioned Exact programs, to make our exposition more uniform.

\subsection{Small proofs of basic counting properties}
Our main results rely on having small proofs of characteristic properties of counting formulae, which we duly give in this section.

First we need to establish a basic monotonicity property:

\begin{proposition}
[$\thr{\vec p}{k}$ is decreasing in $k$]
 \label{monotonicity-of-threshold-subscripts}
 There are polynomial-size $\poselndt$ proofs of the following sequents over extension axioms $\thrextaxs{}{}$:
 \begin{enumerate}
     \item\label{item:thr-0-true} $\seqar \thr {\vec p} 0$
     \item\label{item:thr-k+1-implies-thr-k} $\thr{\vec p}{k+1} \seqar \thr{\vec p}{k}$
     \item\label{item:thr-big-false} $\thr{\vec p}{k} \seqar\  $,  whenever $k>|\vec p|$
 \end{enumerate}
\end{proposition}

\begin{proof}
 We proceed by induction on the length of $\vec p$.
  In the base case, when $\vec p = \epsilon$, all three properties follow easily from $\thrextaxs{}{}$ initial sequents, weakening and cuts.
 
 For the inductive steps, we construct polynomial-size proofs as follows:
 \[
 \begin{array}[b]{rl}
  \text{\cref{item:thr-0-true}}
  \ :
& \ 
 \begin{array}{r@{\ \seqar \ }ll}
    & \thr{\vec p} 0 & \text{by the inductive hypothesis} \\
    & \posdec{\thr{\vec p} 0} p {\thr{\vec p}{-1}} & \text{by \cref{Truth}.\cref{item:truth-0case-implies-posdec} } \\
    & \thr{p\vec p}{0} & \text{by extension axioms $\thrextaxs{}{}$}
\end{array}
 \\
 \noalign{\medskip}
 \text{\cref{item:thr-k+1-implies-thr-k}}
 \ : &\ 
 \begin{array}{r@{\ \seqar\ }ll}
    \thr{p \vec p}{k+1} & \posdec{\thr{\vec p}{k+1}} p {\thr{\vec p} k} & \text{by extension axioms $\thrextaxs{}{}$} \\
    & \posdec{\thr{\vec p}{k}} p {\thr{\vec p} {k-1}} & \text{by inductive hypotheses and \cref{mon-replacing-result}} \\
    & \thr{p \vec p} k & \text{by extension axioms $\thrextaxs{}{}$ again} 
 \end{array}
\\
\noalign{\medskip}
 \text{\cref{item:thr-big-false}}
\ : &\ 
 \begin{array}[b]{r@{\ \seqar \ }ll}
    \thr{p\vec p} k & \posdec{\thr{\vec p}k } p {\thr{\vec p}{k-1}} & \text{by extension axioms $\thrextaxs{}{}$ } \\
    & \thr{\vec p}k, \thr{\vec p}{k-1} & \text{by \cref{Truth}.\cref{item:truth-posdec-implies-1case} } \\
    & \thr{\vec p}{k-1} & \text{by \cref{item:thr-k+1-implies-thr-k} and contraction} \\
    & & \text{by inductive hypothesis} \tag*{\qedhere}
 \end{array}
   \end{array}
 \]
\end{proof}

The arguments above should be read by obtaining each sequent by the justification given on the right, possibly with some cuts and structural rules.

Note that \cref{mon-replacing-result} allows us to apply previously proven implications or equivalences `deeply' within a formula.
We will use such reasoning throughout this work, but shall typically omit further mentioning such uses of \cref{mon-replacing-result} to lighten the exposition.

For the complexity bound, note that only polynomially many lines occur: we only require $k+1$ instantiations of all the sequents above, one for each choice of threshold $i\leq k$.
This sort of complexity analysis will usually suffice for later arguments, in which case we shall suppress them unless further justification is required.

One of the key points we shall exploit in what follows is the provable \emph{symmetry} of $\thr {\vec p} k$, in terms of the ordering of $\vec p$.
We shall establish this through a series of results, beginning by showing a form of `case analysis' on a propositional variable occurring in a list:
\begin{lemma}
[Case analysis]
\label{thresh-case-analysis}
There are polynomial-size $\poselndt$ proofs of,
\[
\thr {\vec p q \vec q} k
\ \dseqar \ 
\thr {q \vec p \vec q} k 
\]
over the extension axioms $\thrextaxs{}{}$.
\end{lemma}
\begin{proof}
We proceed by induction on the length of $\vec p$.
The base case, when $\vec p$ is empty, follows immediately by general identity, \cref{prop:gen-identity}.

For the inductive step we construct polynomial-size proofs  as follows:
\[
\begin{array}[b]{rll}
    & \thr{p \vec p q \vec q} k & \\
    \dseqar & \posdec{\thr{\vec p q \vec q} k } p { \thr{\vec p q \vec q}{k-1} } & \text{by $\thrextaxs{}{}$} \\
    \dseqar & \posdec{\thr{q \vec p  \vec q} k } p { \thr{q \vec p  \vec q}{k-1} } & \text{by $\IH$ and \cref{mon-replacing-result}} \\
    \dseqar & \posdec{ \posdec{\thr{\vec p \vec q} k} q {\thr{\vec p \vec q}{k-1} } }   p { \posdec{\thr{\vec p \vec q}{k-1}} q {\thr{\vec p \vec q}{k-2}} } & \text{by $\thrextaxs{}{}$} \\
    \dseqar & \posdec { \posdec{\thr{\vec p\vec q} k } p { \thr{\vec p \vec q}{k-1} } } q { \posdec{\thr{\vec p \vec q}{k-1}} p { \thr{\vec p \vec q}{k-2} } } & \text{by \cref{pos-medial}} \\
    \dseqar & \posdec{ \thr{p \vec p \vec q} k }  q { \thr{p \vec p \vec q }{k-1} } & \text{by $\thrextaxs{}{}$ }\\
    \dseqar & \thr{q p \vec p\vec q}k & \text{by $\thrextaxs{}{}$} \tag*{\raisebox{\baselineskip}[0pt][0pt]{\qedhere}}
\end{array}
\]
\end{proof}

 Similarly to the proof of \cref{monotonicity-of-threshold-subscripts}, the above argument should be read as providing polynomial-size proofs `in both directions', by the justifications given on the right.
 Note that we restrict cedents to singletons when using $\dseqar$ in this way, to avoid ambiguity of the comma delimiter.
 Polynomial proof size is, again, immediate by inspection on the number of lines.
 
\begin{theorem}
[Symmetry]
\label{thr-symmetric-permutations}
Let $\pi$ be a permutation of $\vec p$. Then there are polynomial-size $\poselndt$ proofs over the extension axioms $\thrextaxs{}{}$ of:
\[
\thr {\vec p} k 
\ \dseqar \
\thr{\pi (\vec p)} k
\]
\end{theorem}

\begin{proof}
Write $\vec p = p_1 \cdots p_n$ and write $\pi(\vec p) = q_1 \cdots q_n$.
We construct polynomial-size proofs by repeatedly applying \cref{thresh-case-analysis} as follows:
\[
\begin{array}{rcll}
    \thr{\vec p} k & \dseqar & \thr{q_n \vec p^n}k & \text{by \cref{thresh-case-analysis}}\\
    & \dseqar &  \thr{q_{n-1}q_n \vec p^{n-1,n}} k & \text{by \cref{thresh-case-analysis}} \\
    & \vdots & & \\
    & \dseqar & \thr{q_1\cdots q_n}k & \text{by \cref{thresh-case-analysis}} \\
    & \dseqar & \thr{\pi(\vec p)}k & \text{by definition of $\vec q$ }
\end{array}
\]

where $\vec p^{i,\dots, n}$ is just $\vec p$ with the elements $q_i, \dots, q_{n}$ removed, otherwise preserving relative order of the propositional variables.
\end{proof}

\section{Case study: the pigeonhole principle}
\label{sec:php-proofs}

The \emph{pigeonhole principle} is usually encoded in propositional logic by a family of $\neg$-free sequents of the following form:
$$
\bigwedge\limits_{i=1}^{n+1} \bigvee\limits_{j=1}^n \pij i j 
\seqar
\bigvee\limits_{j=1}^n \bigvee\limits_{i = 1}^{n} \bigvee\limits_{ i' =  i+1}^{n+1} (\pij i j   \land \pij {i'} j) 
$$
Here it is useful to think of the propositional variables $\pij i j$ as expressing ``pigeon $i$ sits in hole $j$''. In this way the left-hand side (LHS) above expresses that each pigeon $1, \dots, n+1$ sits in some hole $1, \dots , n$, and the right-hand side (RHS) expresses that there is some hole occupied by two (distinct) pigeons.
Note that this encoding allows the mapping from pigeons to holes to be `multi-functional', i.e.\ the LHS allows for a pigeon to sit in multiple holes.

For many propositional proof systems, (quasi)polynomial-size proofs of the pigeonhole principle are obtained by formalising basic properties of counting formulae, as we did in the previous section for $\poselndt$.
This methodology orginated in \cite{DBLP:journals/jsyml/Bussandpigeons87} who showed that Frege admits polynomial-size proofs of the propositional pigeonhole principle.
More pertinent to the present work, \cite{DBLP:journals/mlq/AtseriasGG01} formalised such arguments for quasipolynomial-size Boolean \emph{monotone} counting formulae in the monotone sequent calculus $\mlk$, later improved to a polynomial via \cite{DBLP:journals/jcss/AtseriasGP02,DBLP:journals/apal/Jerabek11a,BKKK17}.

\medskip

In the setting of $\poselndt$ we may not natively express conjunctions, so we adopt a slightly different encoding.
Being a sequent, the outermost conjunctions on the LHS above can simply be replaced by commas;
the subformulas $\pij i j \cand \pij {i'} j$ may be encoded as $\posdec 0 {\pij i j} {\pij {i'} j}$.

Thus we shall work with the following encoding of the pigeonhole principle throughout this section:
\begin{definition}
[Pigeonhole principle]
$\php n$ is the following positive sequent:
\[
\left\{
\bigvee\limits_{j=1}^n \pij i j
\right\}_{i=1}^{n+1}
\seqar
\bigvee\limits_{j=1}^n \bigvee\limits_{i = 1}^{n} \bigvee\limits_{ i' =  i+1}^{n+1} \posdec 0 {\pij i j} {\pij {i'} j}
\]
We write $\lphp n $ and $\rphp n$ for the LHS and RHS, respectively, of $\php n$.
\end{definition}

The main result of this section uses the counting formulas from the previous section to recover small proofs of the pigeonhole principle:
\begin{theorem}
\label{small-proofs-of-php}
There are polynomial-size $\poselndt$ proofs of $\php n$.
\end{theorem}
This result serves somewhat as a warmup, and sanity check, before our more general simulation later in \cref{sec:poselndt-psim-elndt}.
The rest of this section is devoted to its proof.

\subsection{Summary of proof structure}
\label{sect:summary-of-php-proof}
At a high level we shall employ a known proof structure for proving $\php n$, going back to \cite{DBLP:journals/jsyml/Bussandpigeons87} and later employed by \cite{DBLP:journals/mlq/AtseriasGG01}, specialising to our setting for certain intermediate results.
Before surveying this, let us introduce some notation.

\begin{notation}
We fix $n \in \Nat$ throughout this section and write:
\begin{itemize}
    \item $\PPi{i}$ for the list $\pij i 1 , \dots, \pij i n$, and just $\PPi{}$ for the list $\PPi 1 , \dots, \PPi {n+1} $.
    \item $\PPj j$ for the list $\pij 1 j, \dots, \pij{n+1} j $ and just $\PPj{}$ for the list $\PPj 1 , \dots, \PPj n$.
\end{itemize}
\end{notation}
The notation $\PPj{}$ is suggestive since, construing $\PPi{}$ as an $(n+1) \times n$ matrix of propositional variables, $\PPj{}$ is just the transpose $n\times (n+1)$ matrix.

Our approach towards proving $\php n$ in $\poselndt$ (with small proofs) will be broken up into the three smaller steps, proving the following sequents respectively:
\begin{enumerate}
    \item\label{item:lphp-implies-thr-pij} $\lphp n \seqar \thr{\PPi{}}{n+1}$
    \item\label{item:thr-pij-implies-thr-pji} $\thr{\PPi{}}{n+1} \seqar  \thr{\PPj{}}{n+1}$
    \item\label{item:thr-pji-implies-rphp} $\thr{\PPj{}}{n+1} \seqar \rphp n$
\end{enumerate}

Notice that, since $\PPj{}$ is just a permutation of $\PPi{}$, we already have small proofs of \cref{item:thr-pij-implies-thr-pji} from \cref{thr-symmetric-permutations}.
In the next two subsections we shall focus on the other two implications, for which the following lemma will be quite useful:

\begin{lemma}
[Merging and splitting threshold arguments]
\label{further-counting-results-splitting-and-merging}
There are polynomial-size $\poselndt$ proofs, over extension axioms $\thrextaxs{}{}$ of the following sequents:
\begin{enumerate}

    \item \label{merging} $\thr{\vec p } k , \thr {\vec q } l \seqar \thr {\vec p \vec q}{k+l}$
     \item \label{splitting} $ \thr{\vec p \vec q}{k+l} \seqar \thr{\vec p}{k+1}, \thr{\vec q}{l} $
\end{enumerate}
\end{lemma}
\begin{proof}
We proceed by induction on the length of $\vec p$.
In the base case, when $\vec p = \emptylist$, we have two cases for \cref{merging}:
\begin{itemize}
    \item if $k=0$ then $\thr \emptylist 0 , \thr{\vec q} l \seqar \thr {\vec q} l$ follows by $\id$ and $\lefrul \wk$.
    \item if $k\neq 0$ then we have an axiom $\thr \emptylist k \seqar 0$ from $\thrextaxs{}{}$, whence $\thr \emptylist k, \thr{\vec q}l \seqar \thr{\vec q}{k+l}$ follows by $0$, $\cut$ and weakenings.
\end{itemize}
For \cref{splitting}, we have polynomial-size proofs of $\thr{\vec q}{k+l} \seqar \thr{\vec q} l $ already from \cref{monotonicity-of-threshold-subscripts}, whence we obtain $\thr{\vec q}{k+l} \seqar \thr{\emptylist}{k+1}, \thr{\vec q}l$ by $\rigrul{\wk}$.

For the inductive step we shall appeal to \cref{refined-replacement}. 
First, for \cref{merging}, by the inductive hypothesis we already have a polynomial-size proof of:
\[
\begin{array}{r@{ \ \seqar\ }l}
     \thr{\vec p}k, \thr{\vec q}l & \thr{\vec p \vec q }{k+l} \\
     \thr{\vec p}{k-1}, \thr{\vec q} l & \thr{\vec p \vec q}{k+l-1}
\end{array}
\]
Thus, by \cref{refined-replacement} we can derive,
\[
\posdec{\thr{\vec p} k } p {\thr{\vec p}{k-1}} , \thr {\vec q} l \seqar \posdec{\thr{\vec p \vec q}{k+l}} p { \thr{\vec p \vec q }{k+l - 1}  }
\]
whence the required sequent
$ \thr{p \vec p } k , \thr{\vec q} l \seqar \thr{p \vec p \vec q }{k+l}$
follows by the extension axioms $\thrextaxs{}{}$.

For \cref{splitting}, by the inductive hypothesis we have a polynomial-size proof of:
\[
\begin{array}{r@{ \ \seqar\ }l}
     \thr{\vec p \vec q }{k+l} & \thr{\vec p}{k+1}, \thr{\vec q}l \\
     \thr{\vec p \vec q}{k+l-1} &
     \thr{\vec p}{k}, \thr{\vec q} l  
\end{array}
\]
Thus, by \cref{refined-replacement}, we can derive,
\[
\posdec{ \thr{\vec p \vec q}{k+l} } p { \thr{\vec p \vec q}{ k+l-1 } } \seqar \posdec{ \thr{\vec p}{k+1} } p { \thr{\vec p} k } , \thr{\vec q} l
\]
whence the required sequent 
$ \thr{p\vec p \vec q}{k+l} \seqar \thr{p \vec p}{k+1}, \thr{\vec q } l $ follows by the extension axioms $\thrextaxs{}{}$.
\end{proof}

\subsection{From \texorpdfstring{$\lphp n $}{LPHPn} to \texorpdfstring{$(n+1)$}{(n+1)}-threshold}
In this subsection we will give small proofs of the sequent \cref{item:lphp-implies-thr-pij} from \cref{sect:summary-of-php-proof}.

\begin{lemma}
\label{pij-implies-thr-PPi}
Let $\vec q = q_0, \dots, q_{k-1}$.
For all $j<k$,
there are polynomial-size $\poselndt$ proofs over extension axioms $\thrextaxs{}{}$ of:
\[
q_j \seqar \thr{\vec q}{1}
\]
\end{lemma}
\begin{proof}
First we derive,
\begin{equation}
    \label{eq:p-implies-thrp1}
    q_j \dseqar \thr{q_j}1
\end{equation}
as follows:
\[
\begin{array}{r@{\ \dseqar \ }ll}
    q_j & \posdec 0 {q_j } 1 & \text{by \cref{Truth}, axioms and $\cut$ } \\
    & \posdec{\thr \emptylist 1 }{q_j}{\thr\emptylist 0} & \text{by extension axioms $\thrextaxs{}{}$ and \cref{mon-replacing-result}}\\
    & \thr{q_j} 1 & \text{by extension axioms $\thrextaxs{}{}$}
\end{array}
\]
 By repeatedly applying \cref{further-counting-results-splitting-and-merging}.\cref{merging} we obtain polynomial-size proofs of,
\[
    \thr{q_0} 0 , \dots, \thr{q_{j-1}}0, \thr{q_j }1, \thr{q_{j+1} } 0 , \dots, \thr{q_{k-1}} 0 \seqar \thr{\vec q} 1
\]

However, we also have small proofs of $\seqar \thr{q_j}0$ by \cref{monotonicity-of-threshold-subscripts}.\cref{item:thr-0-true}, and so applying $k-1$ cuts we obtain a polynomial-size proof of:
\begin{equation}
\label{eq:thrp1-implies-thrpp1}
    \thr{q_j}1 \seqar \thr{\vec q}1
\end{equation}
The required sequent now follows by simply cutting \cref{eq:p-implies-thrp1} against \cref{eq:thrp1-implies-thrpp1}. 
\end{proof}

\begin{proposition}
\label{lphp-implies-thrPijn+1}
There are polynomial-size $\poselndt$ proofs of \cref{item:lphp-implies-thr-pij}, i.e.,
\[
\lphp n \seqar \thr{\PPi{}}{n+1}
\]
over extension axioms $\thrextaxs{}{}$.
\end{proposition}
\begin{proof}
Let $i \in \{1, \dots, n+1\}$.
By \cref{pij-implies-thr-PPi} above, we have small proofs of,
\[
\pij ij \seqar \thr{\PPi i } 1 
\]
for each $j=1,\dots, n$.
By applying $n-1$ $\lefrul \lor$ steps we derive:
\begin{equation}
    \label{eq:disjunction-implies-thr-1-php}
    \bigvee \PPi i \seqar \thr{\PPi i}1
\end{equation}
Now, applying \cref{further-counting-results-splitting-and-merging}.\cref{merging} $n$ times (and using cuts), we obtain small proofs of:
\begin{equation}
    \label{eq:merging-php}
    \thr{\PPi 1}1, \dots, \thr{\PPi{n+1}} 1 \seqar \thr{\PPi{}}{n+1}
\end{equation}
Finally, by instantiating \cref{eq:disjunction-implies-thr-1-php} for each $i = 1, \dots, n+1$ and applying $n+1$ $\cut$ steps against \cref{eq:merging-php} we derive the required sequent: 
\[
\bigvee \PPi 1, \dots, \bigvee \PPi {n+1} \seqar \thr{\PPi{}}{n+1} \qedhere
\]

\end{proof}

\subsection{From \texorpdfstring{$(n+1)$}{(n+1)}-threshold to \texorpdfstring{$\rphp n$}{RPHPn}}
Before deriving the final sequent \cref{item:thr-pji-implies-rphp} for our proof of $\php n$, we will need some lemmas.
\begin{lemma}
\label{x-and-thry1-implies-disj-x-and-y}
Let $\vec q = q_0, \dots, q_{k-1}$.
There are polynomial-size $\poselndt$ proofs of,
\[
q, \thr{\vec q} 1 \seqar \{ \posdec 0 q {q_i} \}_{i<k}
\]
over extension axioms $\thrextaxs{}{}$.
\end{lemma}
\begin{proof}
For each $i<k$, by \cref{Truth}.\cref{item:truth-1case-implies-posdec} we have a (constant-size) proof of,
\[
q, q_i \seqar \posdec 0 q {q_i}
\]
and so by cutting against appropriate instances of \cref{eq:p-implies-thrp1} we obtain:
\[
    q, \thr{q_i} 1 \seqar \posdec 0 q {q_i}
\]
Instantiating the above for each $i<k$ and applying several $\rigrul \wk$ and $\lefrul \lor$ steps we obtain:
\begin{equation}
\label{eq:x-and-dis-thr-implies-dis-x-and-xi}
    q, \bigvee\limits_{i<k} \thr{q_i} 1 \seqar \{ \posdec 0 q {q_i} \}_{i<k}
\end{equation}

Now, by repeatedly applying \cref{further-counting-results-splitting-and-merging}.\cref{splitting} (under cuts) and $\rigrul \lor$ steps we obtain polynomial-size proofs of:
\begin{equation}
\label{eq:thr1-implies-disj-thr1}
    \thr{\vec q} 1 \seqar \bigvee\limits_{i<k} \thr{q_i}1
\end{equation}
Finally, we conclude by
cutting \cref{eq:thr1-implies-disj-thr1} above against \cref{eq:x-and-dis-thr-implies-dis-x-and-xi}.
\end{proof}

\begin{lemma}
\label{thrx2-implies-disj-x-and-y}
Let $\vec q = q_0, \dots, q_{k-1}$.
There are polynomial-size $\poselndt$ proofs of,
\[
\thr{\vec q} 2 \seqar \{ \posdec 0 {q_i} {q_{i'}} \}_{i<i'<k}
\]
over extension axioms $\thrextaxs{}{}$.
\end{lemma}
\begin{proof}
We proceed by induction on the length $k$ of $\vec q$. In the base case, when $\vec q = \emptylist$, we have an axiom $\thr \emptylist 2 \seqar 0$ from $\thrextaxs{}{}$, whence we conclude by a cut against the $0$-axiom and $\rigrul \wk$.

For the inductive step, we obtain by two applications of \cref{further-counting-results-splitting-and-merging}.\cref{splitting} the following sequents:
\[
\begin{array}{r@{\ \seqar \ }l}
    \thr{q\vec q}{2} & \thr q 1, \thr{\vec q} 2 \\
    \thr{q\vec q}{2} & \thr q 2, \thr{\vec q} 1
\end{array}
\]
Now we already have small proofs of $\thr q 1 \seqar q$ from \cref{eq:p-implies-thrp1} and of $\thr q 2 \seqar $ from \cref{monotonicity-of-threshold-subscripts}.\cref{item:thr-big-false}, and so cutting against the respective sequents above we obtain:

\setlength{\jot}{0pt} 
\begin{align}
    \thr{q\vec q} 2 & \ \seqar\   q,\thr{\vec q} 2 \label{eq:thrxy2-implies-x-or-thry2} \\
    \thr{q\vec q} 2 & \ \seqar\   \thr{\vec q} 1 \label{eq:thrxy2-implies-thry1}
\end{align}

Finally, we combine these sequents using cuts as follows:

\renewcommand{\storageone}{\cref{eq:thrxy2-implies-x-or-thry2}}
\renewcommand{\storagetwo}{\cref{eq:thrxy2-implies-thry1}}
\renewcommand{\storagethree}{\cref{x-and-thry1-implies-disj-x-and-y}}
\[
\vlderivation{
\vliin{\cut}{}{ \thr{q\vec q} 2 \seqar \{ \posdec 0 {q_i} {q_{i'}} \}_{i<i'<k} }{
    \vliiin{2\cut}{}{ \thr{q\vec q} 2 \seqar \{ \posdec 0 {q} {q_{i}}  \}_{i<k}, \thr{\vec q} 2 }{
        \vliq{}{}{ \thr{q\vec q}2 \seqar q,\thr{\vec q} 2  }{\vlhy{\text{\storageone} } }
    }{
        \vliq{}{}{\thr{q\vec q} 2 \seqar \thr{\vec q} 1 }{\vlhy{\text{\storagetwo}}}
    }{
        \vliq{}{}{q,\thr{\vec q} 1 \seqar \{ \posdec 0 {q} {q_{i}}  \}_{i<k} }{ \vlhy{ \text{\storagethree} } }
    }
}{
    \vliq{}{}{ \thr{\vec q} 2 \seqar \{ \posdec 0 {q_i} {q_{i'}} \}_{i<i'<k} }{\vlhy{\IH } }
}
}
\]
where the proof marked $\IH$ is obtained from the inductive hypothesis.
\end{proof}

\begin{proposition}
\label{thrPjin+1-implies-rphp}
There are polynomial-size $\poselndt$ proofs over $\thrextaxs{}{}$ of \cref{item:thr-pji-implies-rphp}, i.e.\ of:
\[
\thr{\PPj{}}{n+1} \seqar \rphp n
\]
\end{proposition}
\begin{proof}
Recall that $\PPj{} = \PPj 1,\dots, \PPj n$, so by $n-1$ applications of \cref{further-counting-results-splitting-and-merging}.\cref{splitting} (and cuts) we have small proofs of:
\begin{equation}
    \label{eq:thrPPn+1-implies-disj-thrPPj2}
    \thr{\PPj{}}{n+1} \seqar \thr{\PPj 1} 2 , \dots, \thr{\PPj n} 2
\end{equation}
Now, instantiating \cref{thrx2-implies-disj-x-and-y} with $\vec q = \PPj j$ we also have small proofs of,
\begin{equation}
    \label{eq:thrPPj2-implies-disj-pij-and-pi'j}
    \thr{\PPj j} 2 \seqar \{\posdec 0 {\pij i j } {\pij {i'} j }  \}_{1\leq i<i'\leq n}
\end{equation}
for each $j=1,\dots, n$. 
Finally, we may apply $n$ $\cut $ steps on \cref{eq:thrPPn+1-implies-disj-thrPPj2} against each instance of \cref{eq:thrPPj2-implies-disj-pij-and-pi'j} (for $j=1,\dots , n$) and apply $\rigrul \lor$ steps to obtain the required sequent. 
\end{proof}

\subsection{Putting it all together}
We are now ready to assemble our proofs for $\php n$.

\begin{proof}
[Proof of \cref{small-proofs-of-php}]
We simply cut together the proofs of \cref{item:lphp-implies-thr-pij}, \cref{item:thr-pij-implies-thr-pji} and \cref{item:thr-pji-implies-rphp} that we have so far constructed:
\renewcommand{\storageone}{\cref{lphp-implies-thrPijn+1}}
\renewcommand{\storagetwo}{\cref{thr-symmetric-permutations}}
\renewcommand{\storagethree}{\cref{thrPjin+1-implies-rphp}}
\[
\vlderivation{
    \vliiin{2\cut}{}{ \lphp n \seqar \rphp n }{
        \vliq{}{}{\lphp n \seqar \thr{\PPi{}}{n+1} }{\vlhy{ \text{\storageone} }}
    }{
        \vliq{}{}{ \thr{\PPi{}}{n+1} \seqar \thr{\PPj{}}{n+1} }{\vlhy{\text{\storagetwo} }}
    }{
        \vliq{}{}{ \thr{\PPj{}}{n+1} \seqar \rphp n }{\vlhy{\text{\storagethree}}}
    }
}
\qedhere
\]
\end{proof}

\section{Positive simulation of non-positive proofs}
\label{sec:poselndt-psim-elndt}
In the previous section we showcased the capacity of the system $\poselndt$ to formalise basic counting arguments by giving polynomial-size proofs of the pigeonhole principle.

In this section we go further and give a general polynomial simulation of $\elndt$, over positive sequents, by adapting a method from \cite{DBLP:journals/jcss/AtseriasGP02}.

\begin{theorem}
\label{poselndt-psims-elndt-pos-seqs}
$\poselndt$ polynomially simulates $\elndt$ over positive sequents.
\end{theorem}
While the high-level structure of the argument is similar to that of \cite{DBLP:journals/jcss/AtseriasGP02}, we must make several specialisations to the current setting due to the peculiarities of eNDT formulas and extension axioms.

\subsection{Summary of proof structure}
Before giving the low-level details, let us survey our approach towards proving \cref{poselndt-psims-elndt-pos-seqs}, in particular comparing it to the analogous methodology from \cite{DBLP:journals/jcss/AtseriasGP02}.

Our strategy is divided into three main parts, which mimic the analogous proof structure from \cite{DBLP:journals/jcss/AtseriasGP02}.

In the first part, \cref{sec:pos-norm-form}, we deal with the non-positive formulas occurring in an $\elndt$ proof. 
The intuition is similar to what is done in \cite{DBLP:journals/jcss/AtseriasGP02} where they first reduced all negations to the variables using De Morgan duality. 
In our setting formulas are no longer closed under duality but, nonetheless, we are able to devise for each formula $A$ an appropriate `positive normal form' $\negtrans A$. 
$\negtrans A$ may contain negative literals (in particular as decision variables), but all decisions themselves are positive.

We duly consider an extension ${\negposelndt}$ of $\poselndt$ which admits negative literals $\dual p$ and has two extra axioms: $p,\dual{p}\seqar $ and $\seqar p,\dual{p}$.

The main result of this first part is that $\negposelndt$ polynomially simulates $\elndt$ over positive sequents (\cref{negposelndt-psims-elndt-pos-seqs}), by essentially replacing each formula occurrence $A$ by $\negtrans A$ and locally repairing the proof (\cref{trans-elndt-to-negposelndt}).

In the second part the aim is to `replace' negative literals in an $\negposelndt$ proof by certain threshold formulas from \cref{Defin.thresh.}. 
This is the same idea as in \cite{DBLP:journals/jcss/AtseriasGP02}, 
but in our setting we must deal with certain technicalities encountered when substituting extended formulas in $\negposelndt$ and $\poselndt$.
In particular, if a literal occurs as a decision variable, then we cannot directly substitute it for an extension variable (e.g.\ a threshold formula $\thr {\vec p} k$), since the syntax of $\elndt$ (and its fragments) does not allow for this.
To handle this issue appropriately, we introduce in \cref{sec:ref-counting-formulas} a refinement of our previous threshold extension variables and axioms, defined mutually inductively with eNDT formulas themselves, that accounts for all such substitution situations (\cref{Thresh.Def}).

For the remainder of the argument, in \cref{sec:subst-thr-for-neg-lits} we fix a $\negposelndt$ proof $P$ of $\Gamma \seqar \Delta$ over extension axioms $\mathcal A$ and propositional variables $\vec p = p_0, \dots, p_{m-1}$. 
We define, for $k\geq 0$, systems $\tposelndt k $ that each have polynomial-size proofs $\ttrans k P$ of $\Gamma \seqar \Delta$  (\cref{k-trans-result}). 
Morally speaking, this simulation is by `substituting' thresholds for negative literals, and the consequent new axioms required in $\tposelndt k$ are parametrised by the threshold $k$.
We point out that $\tposelndt k$ itself is tailored to the specific set of extension axioms $\mathcal A$ and propositional variables $\vec p$ to facilitate the choice of threshold formulas and required extension variables/axioms.

The final part, \cref{sec:main-res-proof}, essentially stitches together proofs obtained in each ${\tposelndt k}$ for $0\leq k \leq m+1$.
More precisely, using basic properties of threshold formulas, we show that each $\tposelndt k$ proof of a positive sequent $\Gamma \seqar \Delta$ can be polynomially transformed into a $\poselndt$ proof of $\thr{\vec{p}}{k},\Gamma \seqar \Delta , \thr{\vec{p}}{k+1}$, over appropriate extension axioms (\cref{simulation:k-true-implies-k+1-true}).
We conclude the argument for our main result \cref{poselndt-psims-elndt-pos-seqs} by simply cutting these together and appealing to \cref{monotonicity-of-threshold-subscripts}.

\subsection{Positive normal form of \texorpdfstring{$\elndt$}{eLNDT} proofs}
\label{sec:pos-norm-form}

We shall temporarily work with a presentation of $\elndt$ within $\poselndt$ by allowing \emph{negative} literals, in order to facilitate our later translations.
For this reason, let us introduce, for each propositional variable $p$, a distinguished propositional variable $\dual p$, which we shall also refer to as `negative literals'.

The system $\negposelndt$ is defined just like $\poselndt$ but also allows negative literals $\dual p$ to appear in (positive) decision steps.
All syntactic positivity constraints remain.
Furthermore, $\negposelndt$ has two additional initial sequents:
\[
\vlinf{\lefrul \neg}{}{p, \dual p \seqar }{}
\qquad
\vlinf{\rigrul \neg}{}{\seqar p, \dual p}{}
\]
The system $\negposelndt$ admits a `normal form' of $\elndt$ proofs:

\begin{definition}
[Positive normal form]
We define a (polynomial-time) translation from an $\elndt$ formula $A$ to a $\negposelndt$ formula $\negtrans A$ as follows:
\[
\begin{array}{r@{\ :=\ }l}
    \negtrans 0 &  0 \\
    \negtrans 1 & 1
\\
    \negtrans p & p\\
    \negtrans{\dual p } & \dual p
\end{array}
\qquad
\begin{array}{r@{\ :=\ }l}
    \negtrans{e_i} & e_i \\
    \noalign{\smallskip}
    \negtrans{(A \lor B)} & \negtrans A \lor \negtrans B\\
    \noalign{\smallskip}
    \negtrans{(\dec A p B ) } & \posdec 0 {\dual p} {\negtrans A} \lor \posdec 0 p {\negtrans B}
\end{array}
\]

For a multiset of formulas $\Gamma = A_1, \dots, A_n$ we write $\negtrans \Gamma := \negtrans A_1, \dots, \negtrans A_n$.

For a set of extension axioms $\mathcal A = \{e_i \extiff A_i\}_{i<n} $, we write $\negtrans {\mathcal A}$ for $\{ e_i \extiff \negtrans A_i \}_{i<n} $.
\end{definition}

We can lift this translation to the level of proofs:

\begin{theorem}
\label{trans-elndt-to-negposelndt}
Let $P$ be an $\elndt$ proof of $\Gamma \seqar \Delta$ over extension axioms $\mathcal A$.
There is a $\negposelndt$ proof $\negtrans P $ of $\negtrans \Gamma \seqar \negtrans \Delta$ over $\negtrans{\mathcal A}$ of size polynomial in $|P|$.
\end{theorem}
\noindent
Notice that, since $\negtrans{\cdot}$ commutes with all connectives except for decisions, to prove the above result it suffices to just derive the translation of decision steps.
For this it will be useful to have another `truth' lemma:

\begin{lemma}
 [Truth for $\negtrans\cdot$-translation]
\label{Truth-negtrans}
Let $\mathcal A=\{e_i \extiff A_i\}_{i<n}$ be a set of positive extension axioms and let $A$ and $B$ be formulas over $e_0, \dots, e_{n-1}$.
There are polynomial-size $\negposelndt$ proofs of the following sequents over $\negtrans {\mathcal A}$:
\begin{enumerate}
    \item\label{item:truth-negtrans-dec-implies-0case} $\negtrans{(\dec A p B)} \seqar                 \negtrans A,p$
    \item\label{item:truth-negtrans-dec-implies-1case} $\negtrans{(\dec A p B)} , p \seqar \negtrans B$
    \item\label{item:truth-negtrans-0case-implies-dec} $\negtrans A \seqar                 \negtrans{(\dec A p B)} , p$
    \item\label{item:truth-negtrans-1case-implies-dec} $p,\negtrans B \seqar \negtrans{(
\dec A p B)}$
\end{enumerate}
\end{lemma}
\begin{proof}
We give the proofs explicitly below:
\[
\vlderivation{
\vliin{\lefrul \lor}{}{ \negtrans{(
\dec A p B)} \seqar \negtrans A, p }{
    \vliin{\lefrul{\pos{\dual p}}}{}{ \posdec 0 {\dual p} {\negtrans A} \seqar \negtrans A , p }{
        \vlin{2\rigrul \wk}{}{0 \seqar \negtrans A, p }{
        \vlin{0}{}{0 \seqar }{\vlhy{}}
        }
    }{
        \vlin{\lefrul \wk, \rigrul \wk}{}{\dual p , \negtrans A \seqar \negtrans A, p}{
        \vlin{\id}{}{\negtrans A \seqar \negtrans A}{\vlhy{}}
        }
    }
}{
    \vliin{\lefrul{\pos{p}} }{}{\posdec 0 p {\negtrans B} \seqar \negtrans A, p }{
        \vlin{2\rigrul\wk}{}{ 0 \seqar \negtrans A,p }{
        \vlin{0}{}{0 \seqar }{\vlhy{}}
        }
    }{
        \vlin{\lefrul \wk, \rigrul \wk}{}{ p, \negtrans B \seqar \negtrans A, p }{
        \vlin{\id}{}{p \seqar p}{\vlhy{}}
        }
    }
}
}
\]
\[
\vlderivation{
\vliin{\lefrul \lor}{}{ \negtrans{(
\dec A p B)}, p \seqar \negtrans B }{
    \vliin{\lefrul{\pos{\dual p}}}{}{ \posdec 0 {\dual p}{\negtrans A} , p \seqar \negtrans B }{
        \vlin{\lefrul \wk, \rigrul \wk}{}{0, p \seqar \negtrans B}{
        \vlin{0}{}{0 \seqar }{\vlhy{}}
        }
    }{
        \vlin{\lefrul\wk, \rigrul \wk}{}{\dual p , \negtrans A, p \seqar \negtrans B }{
        \vlin{\lefrul \neg}{}{\dual p , p \seqar }{\vlhy{}}
        }
    }
}{
    \vliin{\lefrul{\pos p}}{}{ \posdec 0 p {\negtrans B} , p \seqar \negtrans B }{
        \vlin{\lefrul \wk, \rigrul \wk}{}{0 , p \seqar \negtrans B}{
        \vlin{0}{}{0 \seqar }{\vlhy{}}
        }
    }{
        \vlin{2\lefrul \wk}{}{p,\negtrans B, p \seqar \negtrans B}{
        \vlin{\id}{}{\negtrans B\seqar \negtrans B}{\vlhy{}}
        }
    }
}
}
\]
\[
\vlderivation{
\vlin{\rigrul \wk, \rigrul \lor}{}{\negtrans A \seqar \negtrans{(\dec A p B)},p }{
\vliin{\rigrul{\pos {\dual p}}}{}{ \negtrans A \seqar \posdec 0 {\dual p} {\negtrans A},p }{
    \vlin{\lefrul\wk, \rigrul \wk}{}{\negtrans A \seqar 0 ,\dual p , p  }{
    \vlin{\rigrul \neg}{}{\seqar \dual p , p }{\vlhy{}}
    }
}{
    \vlin{2\rigrul\wk}{}{ \negtrans A \seqar 0 , \negtrans A, p }{
    \vlin{\id}{}{\negtrans A \seqar \negtrans A}{\vlhy{}}
    }
}   
}
}
\qquad
\vlderivation{
\vlin{\rigrul \wk, \rigrul \lor}{}{ p,\negtrans B \seqar \negtrans{(\dec A p B)} }{
\vliin{\rigrul{\pos{\dual p}}}{}{p, \negtrans B \seqar \posdec 0 p {\negtrans B} }{
    \vlin{\lefrul \wk , \rigrul \wk }{}{p,\negtrans B \seqar 0 , p}{
    \vlin{\id}{}{p \seqar p}{\vlhy{}}
    }
}{
    \vlin{\lefrul \wk, \rigrul \wk}{}{p, \negtrans B \seqar 0 , \negtrans B}{
    \vlin{\id}{}{\negtrans B \seqar \negtrans B}{\vlhy{}}
    }
}
}
} \qedhere
\]
\end{proof}

We can now prove our polynomial-time interpretation of $\elndt$ within $\negposelndt$:
\begin{proof}
[Proof of \cref{trans-elndt-to-negposelndt}]
We proceed by a straightforward induction on the length of $P$. 
The critical cases are when $P$ ends with decision steps, which we translate as follows.
A left decision step,
\[
\vliinf{\lefrul p}{}{ \Gamma, \dec A p B \seqar \Delta }{ \Gamma , A \seqar \Delta, p }{ \Gamma, p , B \seqar \Delta }
\]
is simulated by the following derivation:
\renewcommand{\storageone}{\cref{Truth-negtrans}.\cref{item:truth-negtrans-dec-implies-0case}}
\renewcommand{\storagetwo}{\cref{Truth-negtrans}.\cref{item:truth-negtrans-dec-implies-1case} }
\[
\vlderivation{
\vliin{\cut}{}{ \negtrans\Gamma, \negtrans{(\dec A p B)} \seqar \negtrans \Delta }{
    \vliin{\cut}{}{ \negtrans \Gamma, \negtrans{(\dec A p B)} \seqar \negtrans \Delta, p }{
        \vliq{}{}{\negtrans{(\dec A p B)} \seqar \negtrans A, p}{\vlhy{\text{\storageone}} }
    }{
        \vlhy{\negtrans\Gamma, \negtrans A \seqar \negtrans \Delta, p }
    }
}{
    \vliin{\cut}{}{\negtrans \Gamma, \negtrans{(\dec A p B) }, p \seqar \negtrans \Delta }{
        \vliq{}{}{\negtrans{(\dec A p B)}, p \seqar \negtrans B }{\vlhy{\text{\storagetwo}}}
    }{
        \vlhy{\negtrans\Gamma, p , \negtrans B \seqar \negtrans\Delta}
    }
}
}
\]
A right decision step,
\[
\vliinf{\rigrul p}{}{\Gamma \seqar \Delta, \dec A p B }{ \Gamma \seqar \Delta, A, p }{ \Gamma , p \seqar \Delta, B }
\]
is simulated by the following derivation:
\renewcommand{\storageone}{\cref{Truth-negtrans}.\cref{item:truth-negtrans-0case-implies-dec}}
\renewcommand{\storagetwo}{\cref{Truth-negtrans}.\cref{item:truth-negtrans-1case-implies-dec}}
\[
\vlderivation{
\vliin{\cut}{}{ \negtrans \Gamma \seqar \negtrans \Delta, \negtrans{(\dec A p B)} }{
    \vliin{\cut}{}{ \negtrans\Gamma \seqar \negtrans\Delta, \negtrans{(\dec A p B)}, p }{
        \vlhy{\negtrans \Gamma \seqar \negtrans \Delta, \negtrans A, p}
    }{
        \vliq{}{}{\negtrans A \seqar \negtrans{(\dec A p B)}, p }{\vlhy{\text{\storageone}}}
    }
}{
    \vliin{\cut}{}{ \negtrans \Gamma, p \seqar \negtrans\Delta, \negtrans{(\dec A p B)} }{
        \vlhy{\negtrans \Gamma, p \seqar \negtrans \Delta, \negtrans B}
    }{
        \vliq{}{}{p,\negtrans B \seqar \negtrans{(\dec A p B)} }{\vlhy{\text{\storagetwo }}}
    }
}
}\qedhere
\]

\end{proof}

Finally we note that the translation above gives rise to a bona fide polynomial simulation of $\elndt$ by $\negposelndt$ over positive sequents:
\begin{corollary}
\label{negposelndt-psims-elndt-pos-seqs}
$\negposelndt$ polynomially simulates $\elndt$, over positive sequents.
\end{corollary}
\begin{proof}
 From \cref{trans-elndt-to-negposelndt} above, it suffices to derive $\Gamma \seqar \Delta$ from $\negtrans \Gamma \seqar \negtrans \Delta$. For this we shall give short proofs of,
 \begin{equation}
     A\ \dseqar \ \negtrans A
 \end{equation}
 when $A$ is positive and free of extension variables, whence $\Gamma \seqar \Delta $ follows from $\negtrans \Gamma \seqar \negtrans \Delta$ by several cuts.
We proceed by structural induction on $A$, for which the critical case is when $A$ is a decision formula.
We prove the two directions separately.

First, note that we have polynomial-size proofs of the following sequents:
\setlength{\jot}{0pt} 
\begin{align}
\posdec A p B & \ \seqar\  A,p & & \text{by \cref{Truth}.\cref{item:truth-posdec-implies-0case}}  \label{eq:posdec-implies-0negtrans} \\
\posdec A p B &\ \seqar\   A \lor  B & & \text{by \cref{Truth}.\cref{item:truth-posdec-implies-1case} and $\rigrul \lor$ }  \label{eq:posdec-implies-1negtrans} \\
A & \ \seqar \   \negtrans{(\posdec A p B )},p & & \text{by \cref{Truth-negtrans}.\cref{item:truth-negtrans-0case-implies-dec} and $\IH$} \label{eq:0case-implies-posdec-negtrans} \\
p,  A \lor  B &\ \seqar\   \negtrans{(\posdec A p B)} & & \text{by \cref{Truth-negtrans}.\cref{item:truth-negtrans-1case-implies-dec} and $\IH$} \label{eq:1case-implies-posdec-negtrans}
\end{align}

We arrange these into a proof of the left-right direction as follows:
\renewcommand{\storageone}{\cref{eq:posdec-implies-0negtrans}}
\renewcommand{\storagetwo}{\cref{eq:0case-implies-posdec-negtrans}}
\renewcommand{\storagethree}{\cref{eq:posdec-implies-1negtrans}}
\renewcommand{\storagefour}{\cref{eq:1case-implies-posdec-negtrans}}
\[
\vlderivation{
\vliin{p\text{-}\cut}{}{ \posdec { A} p { B} \seqar \negtrans{(\posdec A p B)} }{ 
    \vliin{ A\text{-}\cut }{}{\posdec { A} p { B} \seqar \negtrans{(\posdec A p B)} ,p }{
        \vlhy{\text{\storageone} \quad }
    }{
        \vlhy{\quad \text{\storagetwo}}
    }
}{
    \vliin{ A \lor  B\text{-}\cut}{}{ \posdec { A} p { B},p \seqar \negtrans{(\posdec A p B)} }{
        \vlhy{\text{\storagethree}\quad }
    }{
        \vlhy{\quad \text{\storagefour}}
    }
}
}
\]

Next, note that we have small proofs of the following sequents:
\setlength{\jot}{0pt} 
\begin{align}
A  & \ \seqar\ \posdec A p B  & & \text{by \cref{Truth}.\cref{item:truth-0case-implies-posdec}}  \label{eq:A-implies-posdec} \\
p,B  & \ \seqar\ \posdec A p B  & & \text{by \cref{Truth}.\cref{item:truth-1case-implies-posdec}}  \label{eq:xandB-implies-posdec} \\
\negtrans{(\posdec A p B )} & \ \seqar \  {A} , p & & \text{by \cref{Truth-negtrans}.\cref{item:truth-negtrans-dec-implies-0case} and $\IH$} \label{eq:negtransposdec-implies-negtransAorx} \\
 \negtrans{(\posdec A p B)} &\ \seqar\ {A\lor B}  & & \text{by \cref{Truth-negtrans}.\cref{item:truth-negtrans-dec-implies-1case} and $\IH$ } \label{eq:negtransposdecandx-implies-negtransB}
\end{align}

We arrange these into a proof of the left-right direction as follows:
\renewcommand{\storageone}{\cref{eq:A-implies-posdec}}
\renewcommand{\storagetwo}{\cref{eq:xandB-implies-posdec}}
\renewcommand{\storagethree}{\cref{eq:negtransposdec-implies-negtransAorx}}
\renewcommand{\storagefour}{\cref{eq:negtransposdecandx-implies-negtransB}}
\[
\vlderivation{
\vliin{p\text{-}\cut}{}{ \negtrans{\posdec { A} p { B}} \seqar \posdec A p B} { \vliin{A\text{-}\cut}{}{\negtrans{\posdec { A} p { B}} \seqar \posdec A p B,p }{
     \vlhy{ \text{\storagethree} \quad }  }{ \vlhy{\quad \text{\storageone} }}
}{
    \vliin{ A \lor  B\text{-}\cut}{}{\negtrans{\posdec { A} p { B}},p \seqar \posdec A p B }{
     \vlhy{ \text{\storagefour} \quad }  }{
     \vliin{\lefrul \lor}{}{A\lor B, p \seqar \posdec A p B}{\vlhy{\quad \text{\storageone} }}{\vlhy{\quad \text{\storagetwo} }}
     }
}
}\qedhere
\]
\end{proof}

\subsection{Generalised counting formulas}
\label{sec:ref-counting-formulas}
The argument of \cite{DBLP:journals/jcss/AtseriasGP02} relies heavily on substitution of formulas for variables in proofs of $\lk$.
Being based on usual Boolean formulae, this is entirely unproblematic in that setting, whereas in our setting we deal with extension variables that represent NBPs via extension axioms, and so handling substitutions is much more subtle and notationally heavy.

We avoid giving a uniform treatment of this, instead specialising to counting formulas, but we must nonetheless carefully give an appropriate mutually recursive formulation of formulas and extension variables.

\begin{definition}
[Threshold decisions]
\label{Thresh.Def}
We introduce extension variables $\refthr{\vec p}{k}{A}{B}$ for each list $\vec p$ of propositional variables, integer $k$, and formulas $A,B$.

We extend $\thrextaxs{}{}$ to include all extension axioms of the following form, with $\vec p, k,A,B$ ranging as just described:
\begin{equation}
  \label{refined-thresh-axioms}
    \begin{array}{r@{\ \extiff \  }ll}
        \refthr{\epsilon}{0}{A}{B} & A \lor B & \\
        \refthr{\epsilon }{k}{A}{B} & A & k\neq 0 \\
        \refthr{p\vec p}{k}{A}{B} & \posdec{\refthr{\vec p}{k}{A}{B}}{p}{\refthr{\vec p}{k-1}{A}{B} }
    \end{array}
\end{equation}
\end{definition} 

Note that, despite the presentation, $\refthr{\vec p}{k}{A}{B}$ is, formally speaking, a single extension variable, not a decision on the extension variable $\thr {\vec p} k$ which, recall, we do not permit.
This is why we use the square brackets to distinguish it from other formulas, though we shall justify this notation shortly.

One should view the extension variables above and the notion of an $\poselndt$ formula as being \emph{mutually} defined, so as to avoid foundational issues. 
For instance, we allow an extension variable $\refthr{\vec p} k C D$, where $C$ or $D$ may themselves contain extension variables of the form $\refthr{\vec q} k A B$.
By building up formulas and extension variables by mutual induction we ensure that such constructions are well-founded.
Let us briefly make this formal in the following remark:
\begin{remark}
[Well-foundedness of $\thrextaxs{}{}$]
\label{remark-on-well-foundedness-of-refthr}
We may define the extension variables $\refthr {\vec p} k A B$ and $\poselndt$ formulas `in stages' as follows:
\begin{itemize}
    \item Write $\Phi_0$ for the set of all $\poselndt$ formulas over some base set of extension variables $E_0 = \{e_{00}, e_{01}, \dots\}$.
    \item Write $T_n$ for the set of all extension variables $\thr{\vec p} k$ and of the form $\refthr{\vec p} k A B$ with $A,B \in \Phi_n$.
    \item Write $\Phi_{n+1}$ for the set of all formulas built from propositional variables, disjunctions, positive decisions and extension variables from $T_{n}$, $E_n$ and a fresh set of new extension variables $E_{n+1} = \{e_{(n+1)0}, e_{(n+1)1}, \dots \}$.
\end{itemize}
Within each $T_n$ we (partially well-)order extension variables by the length of the superscript $\vec p$, and within each $E_n$ we (well-)order the $e_{ni}$s by the subscript $i$.
Finally, we set $E_n < T_n < E_{n+1}$ (i.e.\, if $e \in E_n$, $ t \in T_n$ and $e' \in E_{n+1}$ then $e<t<e'$). 
In this way the extension axioms from \cref{refined-thresh-axioms} (and \cref{eq:thresh-ext-axs}) indeed satisfy the required well-foundedness criterion. 
We may also admit further extension axioms allowing elements of $E_n$ to abbreviate formulas in $\Phi_n$ (satisfying the indexing condition internal to $E_n$), preserving well-foundedness. 
We shall indeed do this later.

Note, however, that the required order type on indices of these extension variables, a priori, exceeds the ordinal $\omega$.  
This causes no issue for us since, in any finite proof, we will only use finitely many extension variables, and so may construe each index as a (relatively small) natural number while preserving the aforementioned order. 
We shall gloss over this issue in what follows.
\end{remark}

As previously mentioned, our notation $\refthr{\vec p }k A B$ is designed to be suggestive, justified by the following counterpart of the truth conditions from \cref{Truth}:

\begin{proposition}
[Truth conditions for threshold decisions]
\label{Truecond}
\label{Truth-refthr}
There are polynomial size $\poselndt$ proofs over $\thrextaxs{}{}$ of:
\begin{enumerate}
\item \label{item:truth-refthr-posdec-implies-0case}$ \refthr{\vec p}{k}{A}{B} \seqar A , \thr{\vec p}{k} $
    \item \label{item:truth-refthr-posdec-implies-1case}$  \refthr{\vec p}{k}{A}{B} \seqar A , B$
    \item \label{item:truth-refthr-0case-implies-posdec}$A \seqar \refthr{\vec p}{k}{A}{B}$
    \item \label{item:truth-refthr-1case-impliesposdec}$ \thr{\vec p}{k} , B \seqar \refthr{\vec p}{k}{A}{B} $.
\end{enumerate}

\end{proposition}
\begin{proof}
 We proceed by induction on the length of $\vec p$.
For the base case when $\vec p = \emptylist$, we have the following proofs:
\[
\begin{array}{c@{\qquad}c}
    \vlderivation{
    \vlin{\lefrul \wk, \rigrul \wk}{}{\refthr{\emptylist}{0}{A}{B} \seqar A, \thr \emptylist 0  }{
    \vlin{\thrextaxs{}{}, \cut}{}{\seqar \thr \emptylist 0}{
    \vlin{1}{}{\seqar 1}{\vlhy{}}
    }
    }
}
&
\vlderivation{
    \vlin{\id, \lor , \cut}{}{\refthr{\emptylist}{0}{A}{B} \seqar A,B}{
    \vlin{\thrextaxs{}{}}{}{\refthr{\emptylist}{0}{A}{B} \seqar A \lor B}{\vlhy{}}
    }
}
\\
\vlderivation{
    \vlin{\thrextaxs{}{}, \cut}{}{A \seqar \refthr{\emptylist}{0}{A}{B} }{
    \vlin{\rigrul \wk, \rigrul \lor}{}{A \seqar A \lor B}{
    \vlin{\id}{}{A \seqar A}{\vlhy{}}
    }
    }
}
& 
\vlderivation{
    \vlin{\thrextaxs{}{}, \cut}{}{\thr \emptylist 0 , B \seqar \refthr{\emptylist}{0}{A}{B}}{
    \vlin{\lefrul \wk, \rigrul \wk, \rigrul \lor}{}{\thr \emptylist 0 , B \seqar A \lor B}{
    \vlin{\id}{}{B \seqar B}{\vlhy{}}
    }
    }
}
\end{array}
\]

For the inductive step for \cref{item:truth-refthr-posdec-implies-0case} we derive the following sequents:
\[
\begin{array}{r@{\ \seqar \ }ll}
    \refthr{\vec p} k A B & A, \thr{\vec p} k & \text{by $\IH$}  \\
    \refthr{\vec p} {k-1} A B & A, \thr{\vec p}{k-1} & \text{by $\IH$}  \\
    \posdec{\refthr{\vec p} k A B } p { \refthr{\vec p}{k-1} A B } & A, \posdec{\thr{\vec p}k }p{\thr{\vec p} {k-1} } & \text{by \cref{refined-replacement} } \\
    \refthr{p\vec p} k A B & A, \thr{p \vec p} k & \text{by $\thrextaxs{}{}$ and $2\cut$}
\end{array}
\]

For the inductive step for \cref{item:truth-refthr-posdec-implies-1case} we derive the following sequents:
\[
\begin{array}{r@{\ \seqar \ }ll}
    \refthr{\vec p} {k} A B & A,B & \text{by $\IH$}\\
    p,\refthr{\vec p} {k-1} A B & A,B & \text{by $\IH$ and $\lefrul\wk$}\\
    \posdec{\refthr{\vec p} k A B }p{\refthr{\vec p}{k-1} A B } & A,B & \text{by $\lefrul{\pos p}$} \\
    \refthr{p\vec p} k A B & A,B & \text{by $\thrextaxs{}{}$ and $\cut$}
\end{array}
\]

For the inductive step for \cref{item:truth-refthr-0case-implies-posdec} we derive the following sequents:
\[
\begin{array}{r@{\ \seqar \ }ll}
    A & \refthr{\vec p} k A B , p & \text{by $\IH$ and $\rigrul \wk$}\\
    A & \refthr{\vec p} k A B , \refthr{\vec p}{k-1} A B & \text{by $\IH$ and $\rigrul \wk$}\\
    A & \posdec{\refthr{\vec p} k A B }p{\refthr{\vec p}{k-1} A B} & \text{by $\rigrul{\pos p}$ } \\
    A & \refthr{p\vec p} k A B & \text{by $\thrextaxs{}{}$ and $\cut$}
\end{array}
\]

For the inductive step for \cref{item:truth-refthr-1case-impliesposdec} we derive the following sequents:
\[
\begin{array}[b]{r@{\ \seqar \ }ll}
    \thr{\vec p} k , B & \refthr{\vec p} k A B & \text{by $\IH$}\\
    \thr{\vec p}{k-1},B & \refthr{\vec p}{k-1} A B & \text{by $\IH$}\\
    \posdec{ \thr{\vec p} k } p { \thr{\vec p}{k-1} } , B & \posdec{\refthr{\vec p} k A B }p{ \refthr{\vec p}{k-1} A B } & \text{by \cref{refined-replacement}} \\
    \thr{p\vec p}k , B & \refthr{p\vec p} k A B & \text{by $\thrextaxs{}{}$ and $2\cut$} \tag*{\qedhere}
\end{array}
\]
\end{proof}

\subsection{`Substituting' thresholds for negative literals}
\label{sec:subst-thr-for-neg-lits}
For the remainder of this section, let us work with a fixed $\negposelndt$ proof $P$, over extension axioms $\mathcal A = \{ e_i \extiff A_i(e_0, \dots, e_{i-1})  \}_{i<n}$, of a positive sequent $\Gamma \seqar \Delta$ containing propositional variables among $\vec p = p_0, \dots, p_{m-1}$ and their duals $\dual {\vec p} = \dual p_0, \dots, \dual p_{m-1}$, and extension variables among $\vec e = e_0, \dots, e_{n-1}$.

Recall that, since we are eventually trying to give a polynomial simulation of $\elndt$ by $\poselndt$ over positive sequents, our consideration of $\negposelndt$ here suffices, by \cref{negposelndt-psims-elndt-pos-seqs}.
We shall also work with the extension axioms $\thrextaxs{}{}$ from the previous subsection, and will soon explain its interaction with $\mathcal A$ from $P$.

Throughout this section, we shall write $\ppdel i$ for $ p_0, \dots, p_{i-1}  , p_{i+1}, \dots, p_{m-1} $, i.e.\ $\vec p$ with the variable $p_i$ removed. 

We shall define yet another intermediary system $\tposelndt k$, or rather a family of such systems, one for each $k\geq 0$. 
Before that, we need to introduce the following translation of formulas.

\begin{definition}
[`Substituting' thresholds]
We define a (polynomial-time) translation from an $\negposelndt$ formula $A$ (over $\vec p$, $\dual{\vec p}$ and $\vec e$) to an $\poselndt$ formula $\ttrans k A$ (over $\vec p$, some extension variables $\vec e^k$ and extension variables from $\thrextaxs{}{}$) as follows:
\begin{itemize}
\item $\ttrans k 0 := 0$
\item $\ttrans k 1 := 1$
    \item $\ttrans k {p_i} := {p_i}$
    \item $\ttrans k {\dual p_i} :=     \thr{\ppdel i}{k} $
    \item $\ttrans k {e_i } $ is a fresh extension variable. 
    \item $\ttrans k {(A\lor B)} := \ttrans k  A \lor \ttrans k  B$
    \item $\ttrans k {(\posdec A {p_i} B)}:= \posdec {\ttrans k A} {p_i} {\ttrans k B} $
    \item $\ttrans k {(\posdec A {\dual p_i} B )} := \refthr{\ppdel i  }{k}{\ttrans k A}{\ttrans k B} $
\end{itemize}
We also define $\ttrans k {\mathcal A} := \{ \ttrans k {e_i} \extiff \ttrans k {A_i} (\ttrans k {e_0},\dots, \ttrans k {e_{i-1}} )  \}_{i<n} $, 
and if $\Gamma = B_1, \dots, B_l$ we write $\ttrans k \Gamma$ for $\ttrans k {B_1}, \dots, \ttrans k {B_l}$.

\end{definition}

While this translation, and the threshold decisions themselves, may seem syntactically heavy, at the level of branching programs the idea is simple: 
the NBP represented by $\ttrans k A$ is obtained by substituting the NBP represented by $\thr{\ppdel i}{k}$ for each node labelled by $\dual p_i $ in the NBP represented by $A$. 
This is visualised in \cref{fig:blobs}.

 \begin{figure}[t]
$$
A: \quad 
\raisebox{-0.5\height}{
\begin{tikzpicture}[scale=1]

\foreach \pos/\name/\disp in {
  {(.28,3.22)/1/},
  {(-.88,2.88)/2/$\dual p_i$},
  {(.2,2)/3/$\dual p_i$},
  {(.88,.88)/10/},
  {(-.28,3.88)/6/ },
  {(-.08,.88)/8/}}
\node[minimum size=20pt,inner sep=0pt] (\name) at \pos {\disp};

 \draw [->][thick,dotted](1) to (3);
 \draw [->][thin](3) to (10);
    \draw [->][thin](2) to (3);
     \draw
    [->][thick,dotted](3) to (8);
 \path[draw,use Hobby shortcut,closed=true]
(0,0) .. (1,1) .. (1,3) .. (-.8,4) .. (-1.8,2) .. (-1.88,.8) .. (-1.6,.8);

\end{tikzpicture}
}
\qquad \mapsto \qquad 
\ttrans k A : \quad 
\raisebox{-0.5\height}{

\begin{tikzpicture}[scale=1]
 \foreach \pos/\name/\disp in {
  {(5.28,3.22)/4/},
  {(4.12,2.88)/5/},
  {(4.12,2.88)/5/$\thr{\ppdel i}{k}$},
  {(5.2,2)/6/},
  {(4.12,2.88)/12/},
  {(4.68,3.88)/80/ },
  {(4.88,.08)/86/},
  {(4.948,1.42)/84/},
  {(5.948,.28)/85/},
  {(5.12,1.68)/82/$\thr{\ppdel i}{k}$}}
\node[minimum size=20pt,inner sep=0pt] (\name) at \pos {\disp};
  \draw [->][thick,dotted](4) to (6);
    \draw [->][thin](5) to (6);
     \draw [->][thick,dotted](84) to (86);
      \draw [->][thin](84) to (85);
  
 \path[draw,use Hobby shortcut,closed=true]
(5,0) .. (6,1) .. (6,3) .. (4.2,4) .. (3.2,2) .. (3.12,.8) .. (3.4,.8);

\path[draw,use Hobby shortcut,closed=true]
(4.98,2.2) .. (5.62,2.2) .. (5.68,1.88) .. (4.68,1.4) ;

\path[draw,use Hobby shortcut,closed=true]
(4,3.4) .. (4.64,3.4) .. (4.7,3.08) .. (3.7,2.7) ;
\end{tikzpicture}
}
$$
\caption{Visualising the $\ttrans k \cdot$-translation: the NBP represented by $\ttrans k A$ is obtained by substituting $\thr{\ppdel i}{k}$ for $\dual p_i$ in the NBP represented by $A$.}
\label{fig:blobs}
\end{figure}

In what follows, we shall work with the set of extension axioms $\thrextaxs{}{} \cup \ttrans k {\mathcal A}$, so let us take a moment to justify that this set of extension axioms is indeed well-founded.
\begin{remark}
[Well-foundedness, again]
Following on from \cref{remark-on-well-foundedness-of-refthr}, 
well-foundedness of $\thrextaxs{}{} \cup \ttrans k {\mathcal A}$ follows from a suitable indexing of the extension variables therein. 
To this end we assign `stages' to each formula $\ttrans k A$ by $\mathcal A$-induction on $A$, using the notation of \cref{remark-on-well-foundedness-of-refthr}:
\begin{itemize}
    \item $\ttrans k {p_i}, \ttrans k 0 , \ttrans k 1  \in \Phi_0$.
    \item $\ttrans k {\dual p_i}  \in T_0$.
    \item $\ttrans k {e_i} \in E_{m} \subseteq \Phi_m$ if $\ttrans k {A_i} (\ttrans k {e_0}, \dots, \ttrans k {e_{i-1}} ) \in \Phi_m$, with index $i$ (i.e., $\ttrans k {e_i}$ is $e_{mi}$).
    \item $\ttrans k {(A \lor B)} \in \Phi_m$ if $\ttrans k A ,\ttrans k B \in \Phi_m$.
    \item $\ttrans k {(\posdec A {p_i} B)} \in \Phi_m $ if $\ttrans k A ,\ttrans k B \in \Phi_m$.
    \item $\ttrans k {(\posdec A {\dual p_i} B )} \in T_{m} \subseteq \Phi_{m+1}$ if $\ttrans k A , \ttrans k B \in \Phi_m$.
\end{itemize}
Note in particular that stages can grow even for formulas free of $\ttrans k {e_i}$ since we may have nested decisions on negated variables $\dual p_i$.

Once again, in terms of proof complexity, we will gloss over this subtlety and simply count the number of propositional and extension variable occurrences in a proof, assuming that each variable can be equipped with a `small' index.
\end{remark}

We are now ready to define our intermediary systems.

\begin{definition}
\label{tposelndt-definition}
The system $\tposelndt k$ is defined just like $\poselndt$, but includes additional initial sequents,
\begin{equation}
    \label{eq:threshold-initial-sequents}
    \vlinf{\lefrul t}{}{p_i, \thr{\ppdel i}{k} \seqar }{}
\qquad
\vlinf{\rigrul t}{}{\seqar p_i, \thr{\ppdel i}{k}  }{}
\end{equation}
and may only use the extension axioms $\thrextaxs{}{} \cup \ttrans k {\mathcal A}$.
\end{definition}

\begin{lemma}
\label{k-trans-result}
There is an $\tposelndt k$ proof of $   \Gamma \seqar  \Delta $ of size polynomial in $|P|$.
\end{lemma}
\begin{proof}
We construct the required proof $\ttrans k P$ by replacing every formula occurrence $A$ in $P$ by $\ttrans k A$. Note that all structural steps, identities, cuts remain correct.
An extension axiom for $e_i$ from $\mathcal A$ is just translated to the corresponding extension axiom for $\ttrans k {e_i}$ from $\ttrans k {\mathcal A}$, and the initial sequents $\lefrul \neg$ and $\rigrul \neg$ from $\negposelndt$ are translated to the two new initial sequents $\lefrul t$ and $\rigrul t$, respectively, from 
\cref{eq:threshold-initial-sequents} above.
It remains to simulate the logical steps.

The simulation of $\lor$ steps is immediate, since the $\ttrans k \cdot$-translation commutes with $\lor$.
Similarly for positive decisions on $p_i$.
A left positive decision step on $\dual p_i$,
\[
\vliinf{\lefrul{\pos{\dual p_i}}}{}{ \Gamma, \posdec A {\dual p_i} B \seqar \Delta }{\Gamma, A \seqar \Delta }{\Gamma, \dual p_i, B \seqar \Delta}
\]
is translated to the following derivation:
\renewcommand{\storageone}{\cref{Truth-refthr}.\cref{item:truth-refthr-posdec-implies-0case}}
\renewcommand{\storagetwo}{\cref{Truth-refthr}.\cref{item:truth-refthr-posdec-implies-1case}}

\[
\small
\vlderivation{
    \vliin{\cut}{}{ \ttrans k \Gamma, \ttrans k {(\posdec A {\dual p_i} B)} \seqar \ttrans k \Delta  }{
        \vliin{\cut}{}{\ttrans k \Gamma, \ttrans k {(\posdec A {\dual p_i} B)} \seqar \ttrans k \Delta, \ttrans k A }{
            \vliq{}{}{\ttrans k {(\posdec A {\dual p_i} B)} \seqar \ttrans k A, \ttrans k {\dual p_i} }{\vlhy{\text{\storageone}}}
        }{
            \vliin{\cut}{}{ \ttrans k \Gamma, \ttrans k {(\posdec A {\dual p_i} B)} , \ttrans k {\dual p_i} \seqar \ttrans k \Delta, \ttrans k A }{
                \vliq{}{}{ \ttrans k {(\posdec A {\dual p_i} B)} \seqar \ttrans k A , \ttrans k B }{ \vlhy{\text{\storagetwo}} }
            }{
                \vlhy{\ttrans k \Gamma, \ttrans k {\dual p_i}, \ttrans k B \seqar \ttrans k \Delta}
            }
        }
    }{
        \vlhy{\ttrans k \Gamma, \ttrans k A \seqar \ttrans k \Delta}
    }
}
\]
A right positive decision rule on $\dual p_i$,
\renewcommand{\storageone}{\cref{Truth-refthr}.\cref{item:truth-refthr-1case-impliesposdec}}
\renewcommand{\storagetwo}{\cref{Truth-refthr}.\cref{item:truth-refthr-0case-implies-posdec}}
\[
\vliinf{\rigrul{\pos{\dual p_i}}}{}{ \Gamma \seqar \Delta, \posdec A {\dual p_i} B }{
    \Gamma \seqar \Delta, A, \dual p_i
}{
    \Gamma \seqar \Delta, A, B
}
\]
is translated to the following derivation:
\[
\small
\vlderivation{
\vliin{\cut}{}{\ttrans k \Gamma \seqar \ttrans k \Delta, \ttrans k {\posdec A {\dual p_i} B}}{
    \vliin{\cut}{}{ \ttrans k \Gamma \seqar \ttrans k \Delta, \ttrans k A, \ttrans k {\posdec A {\dual p_i} B} }{
        \vlhy{\ttrans k \Gamma \seqar \ttrans k \Delta, \ttrans k A, \ttrans k {\dual p_i} }
    }{
        \vliin{\cut}{}{ \ttrans k \Gamma , \ttrans k {\dual p_i} \seqar \ttrans k \Delta , \ttrans k A, \ttrans k {\posdec A {\dual p_i} B}  }{
            \vlhy{\ttrans k \Gamma \seqar \ttrans k \Delta, \ttrans k A, \ttrans k B}
        }{
            \vliq{}{}{ \ttrans k {\dual p_i}, \ttrans k B \seqar \ttrans k {\posdec A {\dual p_i} B} }{\vlhy{\text{\storageone}}}
        }
    }
}{
    \vliq{}{}{\ttrans k A \seqar \ttrans k {\posdec A {\dual p_i} B}}{\vlhy{\text{\storagetwo}}}
}
}
\qedhere
\]
\end{proof}

\subsection{Putting it all together}
\label{sec:main-res-proof}
We are now ready to assemble the proof our main simulation result, \cref{poselndt-psims-elndt-pos-seqs}.
Recall that we are still working with the fixed $\negposelndt$ proof $P$ of a positive sequent $\Gamma \seqar \Delta$ from \cref{sec:subst-thr-for-neg-lits}, over extension axioms $\mathcal A = \{ e_i \extiff A_i\}_{i<n} $ and propositional variables $\vec p = p_0, \dots, p_{m-1}$.
We continue to write $\ppdel i$ for $p_0, \dots, p_{i-1},p_{i+1}, \dots, p_{m-1}$, i.e.\ $\vec p$ with $p_i$ removed.

\begin{proposition}
\label{thresh-increment}
For $k\geq 0$, there are polynomial size $\poselndt$ proofs of,

\setlength{\jot}{0pt}
\begin{align}
    p_i , \thr {\ppdel i } k  \seqar & \thr{\vec p}{k+1} \label{eq:thresh-inctrement-left} \\
    \thr{\vec p} k \seqar & p_i , \thr{\ppdel i} k \label{eq:thresh-increment-right}
\end{align}
over extension axioms $\thrextaxs{}{}$.
\end{proposition}
\begin{proof}

We derive \cref{eq:thresh-inctrement-left} as follows:
\[
\begin{array}{r@{\ \seqar \ }ll}
    \thr{p_i} 1 , \thr{\ppdel i}{k+1} & \thr{p_i\ppdel i}{k+1} &  \text{by \cref{further-counting-results-splitting-and-merging}.\cref{merging} } \\
    p_i, \thr{\ppdel i }k & \thr{p_i\ppdel i}{k+1}  &\text{by \cref{eq:p-implies-thrp1} and $\cut$} \\
     & \thr{\vec p}{k+1} & \text{by \cref{thresh-case-analysis} and $\cut$}
\end{array}
\]
We derive \cref{eq:thresh-increment-right} as follows:
\[
\begin{array}[b]{r@{\ \seqar \ }ll}
    \thr{\vec p} k & \thr{p_i \ppdel i} k & \text{by \cref{thresh-case-analysis}} \\
    & \thr{p_i} 1 , \thr{\ppdel i} k & \text{by \cref{further-counting-results-splitting-and-merging}.\cref{splitting} and $\cut$} \\
    &  p_i ,\thr{\ppdel i}k & \text{by \cref{eq:p-implies-thrp1} and $\cut$} \tag*{\qedhere}
\end{array}
\]
\end{proof}

\begin{lemma}
\label{simulation:k-true-implies-k+1-true}
For $k\geq 0$, there are polynomial size $\poselndt$ proofs of,
$$\thr {\vec p} k , \Gamma \seqar \Delta, \thr {\vec p} {k+1}$$ 
over extension axioms $\thrextaxs{}{} \cup \ttrans k {\mathcal A} $.
\end{lemma}
\begin{proof}
By \cref{k-trans-result}, we already have a polynomial-size proof $\tposelndt k$ proof $\ttrans k P$ of $\Gamma \seqar \Delta$.
By definition of $\tposelndt k$, we construe $\ttrans k P$ as an $\poselndt$ derivation of $\Gamma \seqar \Delta$ over extension axioms $ \thrextaxs{}{} \cup \ttrans k {\mathcal A}$ from hypotheses:
\setlength{\jot}{0pt}
\begin{align}
    p_i, \thr{\ppdel i} k \seqar & \label{eq:thr-id-hyp-left}  \\
    \seqar & p_i , \thr{\ppdel i} k \label{eq:thr-id-hyp-right}
\end{align}
We obtain the required proof by adding $\thr{\vec p} k$ to the LHS of each sequent and $\thr{\vec p}{k+1}$ to the RHS of each sequent in $\ttrans k P$.
Each local inference step remains correct, except that some weakenings may be required to repair initial steps.
Finally we replace occurrences of the hypotheses \cref{eq:thr-id-hyp-left} and \cref{eq:thr-id-hyp-right} above by the proofs of \cref{eq:thresh-inctrement-left} and \cref{eq:thresh-increment-right} respectively from \cref{thresh-increment}.
\end{proof}

We are now ready to prove our main result, that $\poselndt$ polynomially simulates $\elndt$ over positive sequents:
\begin{proof}
[Proof of \cref{poselndt-psims-elndt-pos-seqs} ]
By \cref{negposelndt-psims-elndt-pos-seqs}, without loss of generality let $P$ be an $\negposelndt$ proof of a positive sequent $\Gamma \seqar \Delta$ over extension axioms $\mathcal A$.
By \cref{simulation:k-true-implies-k+1-true} we construct, for each $k\leq m+1$, polynomial-size proofs of $\thr {\vec p} k, \Gamma \seqar \Delta, \thr{\vec p}{k+1}$, over $\thrextaxs{}{} \cup \ttrans k {\mathcal A}$, and we simply `cut' them all together as follows:
\renewcommand{\storageone}{\cref{monotonicity-of-threshold-subscripts}.\cref{item:thr-0-true}}
\renewcommand{\storagetwo}{\cref{simulation:k-true-implies-k+1-true}}
\renewcommand{\storagethree}{\cref{monotonicity-of-threshold-subscripts}.\cref{item:thr-big-false}}
\[
\vlderivation{
    \vliiiiin{(m+2)\cut}{}{\Gamma \seqar \Delta}{
        \vliq{}{}{\seqar \thr{\vec p} 0 }{\vlhy{\text{\storageone}}}
    }
    {
        \vliq{}{}{\thr {\vec p} 0 , \Gamma \seqar \Delta, \thr {\vec p} 1}{\vlhy{\text{\storagetwo}}}
    }
    {
        \vlhy{  \cdots  }
    }
    {
        \vliq{}{}{ \thr {\vec p} m , \Gamma \seqar \Delta, \thr {\vec p} {m+1}}{\vlhy{\text{\storagetwo}}}
    }
    {
        \vliq{}{}{ \thr{\vec p}{m+1} \seqar  }{\vlhy{\text{\storagethree}}}
    }
}
\]
The resulting proof is an $\poselndt$ proof of the required sequent, over extension axioms $\thrextaxs{}{} \cup \ttrans 0 {\mathcal A} \cup \ttrans 1 {\mathcal A} \cup \cdots \cup \ttrans {m+1} {\mathcal A} $.
Note that this set of extension axioms is indeed well-founded, since each $\ttrans k {\mathcal A}$ is only used in distinct subproofs.
\end{proof}

\section{Conclusions}
In this work we studied the proof complexity of \emph{positive} non-deterministic branching programs, where each $0$-transition has a parallel $1$-transition.
We built upon foundational work introducing a `Frege' system $\elndt$ for (non-deterministic) branching programs \cite{DBLP:conf/csl/BussDasKnop20}, recovering its `positive' fragment $\poselndt$.

Inspired by previous upper bounds in proof complexity via counting arguments, e.g.\ \cite{DBLP:journals/jsyml/Bussandpigeons87,DBLP:journals/mlq/AtseriasGG01,DBLP:journals/jcss/AtseriasGP02}, we constructed polynomial-size proofs of basic counting principles in $\poselndt$, over known positive branching programs for counting.
We applied these principles to explicitly construct polynomial-size proofs of the propositional pigeonhole principle, and then went further to give a general polynomial simulation of $\elndt$ by $\poselndt$ over positive sequents.
Note that the former result is indeed subsumed by the latter, albeit somewhat indirectly: \cite{DBLP:conf/csl/BussDasKnop20} showed that $\elndt $ polynomially simulates Frege systems, and \cite{DBLP:journals/jsyml/Bussandpigeons87} gave polynomial-size Frege proofs of the pigeonhole principle.

As we have mentioned, the positive fragment $\mlk$ of usual Frege calculi, based on Boolean formulas, is well studied, and we now know that $\mlk$ polynomially simulates $\lk$ over positive sequents \cite{DBLP:journals/jcss/AtseriasGP02,DBLP:journals/apal/Jerabek11a,BKKK17}.
This solves a proof complexity theoretic version of the $\mon\NC 1 $ vs $\NC 1$ question.
Furthermore Je{\v r}{\'a}bek has obtained an analogous result for monotone circuits \cite{jerabek:monotone-cuts}, solving the proof complexity theoretic version of the $\mon\Ptime$ vs $\Ptime$ question.
In this sense, this paper is a contribution to this line of work, this time solving the $\mon\NL$ vs $\NL$ question.
Note that all these questions are solved positively, in the sense that positive systems polynomially simulate their non-positive counterparts (over positive sequents).
This is in contrast to the world of circuit complexity, where the classes $\NC 1 , \NL , \Ptime$ are respectively known to be separated from $\mon\NC1, \mon \NL, \mon \Ptime$,   by \cite{RazWig90:mNC1-vs-NC1,GrigniSnipser,razborov1985} respectively.

\medskip

There are still several open questions in the area of monotone proof complexity.
Notably, the aforementioned polynomial simulation of $\lk$ by $\mlk$ holds, a priori, only for the dag-like version of $\mlk$.
While a quasipolynomial simulation for the tree-like version is known from \cite{DBLP:journals/jcss/AtseriasGP02}, it remains open whether tree-$\mlk$ polynomially simulates $\lk$ over positive sequents.
One could pose a similar question for tree-like versions of the systems $\elndt,\poselndt$ studied in this work, but it seems less robust to do so as these calculi are based on extension axioms which themselves represent dag-like objects, seemingly requiring dag-like proofs to efficiently prove even basic manipulations, such as logical equivalence of isomorphic representations of branching programs.

It would also be interesting to study monotone complexity at the level of \emph{bounded arithmetic} or \emph{equational theories}. 
We now have such theories corresponding to $\NC 1 , \NL , \Ptime$ (see, e.g., \cite{CooNgu10:log-found-prf-comp}), and it would be interesting to develop analogous theories for positive classes and proof systems. 
Let us point out that  there has been some previous work in this direction, \cite{Das16:pos-int-bdd-arith}, where certain positive and inuitionistic theories were shown to enjoy quasipolynomial-time translations to certain positive proof systems.

\bibliographystyle{alphaurl}
\bibliography{bib}

\end{document}